\documentclass{article}
\usepackage{graphicx} 

\usepackage{amssymb}
\usepackage{amsmath}
\usepackage{amsthm}
\usepackage{fullpage}
\usepackage{latexsym}
\usepackage{graphicx}
\usepackage{color}
\usepackage{graphics}
\usepackage[round]{natbib}
\usepackage{cancel}
\usepackage{thm-restate}
\usepackage{mathtools}
\usepackage[mathscr]{euscript}
\usepackage{hyperref}
\usepackage{tikz,wrapfig}

\usepackage{enumitem}

\usepackage[algo2e,ruled]{algorithm2e}
\usepackage{algorithm}
\usepackage[noend]{algpseudocode}

\DeclareMathAlphabet\mathbb{U}{msb}{m}{n}
\usepackage{xpatch}

\def\Rset{\mathbb{R}}

\DeclareMathOperator*{\E}{\mathbb E}

\DeclareMathOperator*{\conv}{conv}

\DeclareMathOperator{\poly}{poly}
\DeclareMathOperator{\diam}{diam}

\DeclarePairedDelimiter{\norm}{\|}{\|}

\newcommand{\cA}{\mathcal{A}}
\newcommand{\cB}{\mathcal{B}}
\newcommand{\cC}{\mathcal{C}}

\newcommand{\cK}{\mathcal{K}}

\newcommand{\cM}{\mathcal{M}}

\newcommand{\cP}{\mathcal{P}}

\newcommand{\cS}{\mathcal{S}}

\newcommand{\cU}{\mathcal{U}}
\newcommand{\cV}{\mathcal{V}}

\newcommand{\cX}{\mathcal{X}}
\newcommand{\cY}{\mathcal{Y}}

\newcommand{\bx}{{\mathbf x}}
\newcommand{\by}{{\mathbf y}}
\newcommand{\br}{{\mathbf r}}

\newcommand{\bu}{{\mathbf u}}

\newcommand{\eps}{\varepsilon}
\newcommand{\ignore}[1]{}

\newcommand{\yV}{y^{V}}
\newcommand{\xV}{x^{V}}
\newcommand{\byV}{\by^{V}}
\newcommand{\bxV}{\bx^{V}}

\DeclareMathOperator{\BR}{\mathsf{BR}}
\DeclareMathOperator{\Stack}{\mathsf{Stack}}
\DeclareMathOperator{\Reg}{\mathsf{Reg}}
\DeclareMathOperator{\Swap}{\mathsf{SwapReg}}
\DeclareMathOperator{\SwapDist}{\mathsf{SwapDist}}
\DeclareMathOperator{\LinSwap}{\mathsf{LinSwapReg}}
\DeclareMathOperator{\CorrSwap}{\mathsf{ProfSwapReg}}
\DeclareMathOperator{\CorrDist}{\mathsf{ProfSwapDist}}
\DeclareMathOperator{\PolySwap}{\mathsf{PolySwapReg}}
\DeclareMathOperator{\NormSwap}{\mathsf{NFSwapReg}}
\DeclareMathOperator{\AppLoss}{\mathsf{AppLoss}}

\DeclareMathOperator{\Proj}{\mathrm{Proj}}
\DeclareMathOperator{\ShellProj}{\mathsf{ShellProj}}

\usepackage{comment}
\usepackage{bm}

\usepackage{blkarray}
\newtheorem{theorem}{Theorem}

\newtheorem{lemma}[theorem]{Lemma}

\theoremstyle{remark}
\newtheorem{remark}{Remark}

\newcommand{\optset}{\cY}
\newcommand{\learnset}{\cX}
\newcommand{\vset}{\cV}
\newcommand{\learnvert}{\vset(\learnset)}
\newcommand{\optvert}{\vset(\optset)}
\newcommand{\csp}{\phi}
\newcommand{\cspV}{\phi^{V}}
\newcommand{\ball}{\cB}

\title{Swap Regret and Correlated Equilibria Beyond Normal-Form~Games}
\author{Eshwar Ram Arunachaleswaran\thanks{University of Pennsylvania, \texttt{eshwarram.arunachaleswaran@gmail.com}.  }
\and 
Natalie Collina\thanks{University of Pennsylvania, \texttt{ncollina@seas.upenn.edu} } \and Yishay Mansour \thanks{Tel Aviv University and  Google Research, {\tt mansour.yishay@gmail.com}. 
} 
\and Mehryar Mohri \thanks{Google Research and Courant Institute of Mathematical Sciences, New York, {\tt mohri@google.com}} 
\and Jon Schneider \thanks{Google Research, {\tt jschnei@google.com}} 
\and Balasubramanian Sivan \thanks{Google Research, {\tt balusivan@google.com}}}

\begin{document}

\maketitle

\begin{abstract}
Swap regret is a notion that has proven itself to be central to the study of general-sum normal-form games, with swap-regret minimization leading to convergence to the set of correlated equilibria and guaranteeing non-manipulability against a self-interested opponent. However, the situation for more general classes of games -- such as Bayesian games and extensive-form games -- is less clear-cut, with multiple candidate definitions for swap-regret but no known efficiently minimizable variant of swap regret that implies analogous non-manipulability guarantees. 

In this paper, we present a new variant of swap regret for polytope games that we call ``profile swap regret'', with the property that obtaining sublinear profile swap regret is both necessary and sufficient for any learning algorithm to be non-manipulable by an opponent (resolving an open problem of Mansour et al., 2022). Although we show profile swap regret is NP-hard to compute given a transcript of play, we show it is nonetheless possible to design efficient learning algorithms that guarantee at most $O(\sqrt{T})$ profile swap regret. Finally, we explore the correlated equilibrium notion induced by low-profile-swap-regret play, and demonstrate a gap between the set of outcomes that can be implemented by this learning process and the set of outcomes that can be implemented by a third-party mediator (in contrast to the situation in normal-form games).
\end{abstract}

\section{Introduction}

The theory of regret minimization is one of our most powerful tools for understanding zero-sum games. Consider, for instance, the statement that two players, choosing their strategies in a repeated zero-sum game by running a no-regret learning algorithm, each asymptotically guarantee that they receive their minimax equilibrium value \citep{FreundSchapire1996}. This simple fact lies at the heart of many theoretical results across computer science and economics, and underlies many of the recent super-human level AI performances in games such as Poker and Go \citep{BrownSandholm2017,MoravcikSBLMBDW2017,SilverHMGSDSAPL2016}. 

Many important games are not zero-sum, but instead \emph{general-sum}, allowing for potential alignment between the incentives of the different players. Regret minimization is also a powerful tool for analyzing general-sum games, albeit with some additional caveats. Minimizing the standard notion of \emph{external regret} -- the gap between one's utility and the utility of the best fixed action in hindsight -- has markedly weaker guarantees than in the zero-sum setting (this only allows convergence to the considerably weaker class of coarse-correlated equilibria, and is vulnerable to certain manipulations by a strategic opponent \citep{deng2019strategizing}). In these cases, one can get stronger game-theoretic guarantees by minimizing \emph{swap regret} -- the gap between one's utility and the counterfactual utility one could have received had they applied the best static \emph{swap function} to all of their actions (e.g., playing rock everywhere they might have previously played scissors). It is a fundamental result in online learning \citep{blum2007external} that there exist efficient swap regret minimization algorithms for playing in general-sum, normal-form games. 

In addition to being general-sum, many games come with some additional structure. For example, in Bayesian games (capturing settings like auctions and markets) players have private information they can use to choose their actions, and in extensive-form games (capturing settings like Poker and bargaining) players take multiple actions in sequence. Both of these classes of games fall under the umbrella of \emph{polytope games}\footnote{These games (or slight variants thereof) also appear in the literature under the names \emph{convex games} \citep{daskalakis2024efficient} and \emph{polyhedral games} \citep{Farina22:NearOptimal}. We choose the term \emph{polytope games} to be consistent with \cite{MMSSbayesian}, whose work we most directly build off of.}: (two-player\footnote{Throughout this paper, we will restrict our attention to games with two players.}) games where each player takes actions in some convex polytope, and where the utilities are provided as bilinear functions of these two action vectors. 

Ideally, we would like to apply the theory of regret minimization by minimizing swap regret in the class of polytope games, thereby taking advantage of the previously-mentioned guarantees. It is here where we hit a stumbling block -- for reasons we will discuss shortly, there is no single definition of swap regret in polytope games, but instead a wealth of different definitions: linear swap regret, polytope swap regret, normal-form swap regret, $\Phi$-regret, low-degree swap regret, etc. \citep{zhang2024efficient, fujii2023bayes, MMSSbayesian, GordonGreenwaldMarks2008}.  Of these notions, the strongest ones which do provide some of the above guarantees (e.g., normal-form swap regret) are conjectured to be intractable to minimize, and the ones for which we have efficient learning algorithms (e.g., linear-swap regret) come with unclear game-theoretic guarantees.

\subsection{Main Results and Techniques}

Our main contribution in this paper is the introduction of a new variant of swap regret in polytope games which we call \emph{profile swap regret}. We argue that profile swap regret is a particularly natural notion of swap regret in polytope games for two primary reasons:

\begin{itemize}
\item First, many of the utility-theoretic properties that swap regret possesses in general-sum normal-form games extend naturally to profile swap regret in polytope games. Take as one example the property of \emph{non-manipulability}. In \cite{deng2019strategizing}, the authors noticed that common no-external-regret algorithms (such as multiplicative weights) have the property that a strategic opponent can manipulate the play of these learning algorithms to achieve asymptotically more utility than they would be able to achieve against a rational follower (i.e., the ``Stackelberg value'' of the game). In contrast, no-swap-regret algorithms cannot be manipulated in this way, and it is now understood that the property of incurring sublinear swap regret exactly characterizes when learning algorithms are non-manipulable in general-sum normal-form games \cite{deng2019strategizing, MMSSbayesian}.

Similarly, we show that the property of incurring sublinear profile swap regret is equivalent to the property of being non-manipulable in polytope games (Theorem \ref{thm:poly_nonmanip}), resolving (one interpretation of) an open question of \cite{MMSSbayesian}. We additionally show that no-profile-swap-regret algorithms have minimal asymptotic menus and are Pareto-optimal, two properties established by \cite{paretooptimal} for no-swap-regret algorithms in normal-form games (Theorems \ref{thm:poly_minimal} and \ref{thm:poly_pareto}).

\item Secondly, by extending a recent technique introduced in \cite{daskalakis2024efficient}, we show that it is possible to design efficient learning algorithms that guarantee sublinear profile swap regret. That is, we show it is possible to construct explicit learning algorithms that run in time polynomial in the size of the game and number of rounds which incur at most $O(\sqrt{T})$ profile swap regret (Theorem~\ref{thm:upper-semi-separation}). Perhaps more meaningfully, these algorithms guarantee that any opponent can gain at most $O(\sqrt{dT})$ utility over their Stackelberg value in a $d$-dimensional polytope game by attempting to manipulate this learning algorithm, and are the first known efficient algorithms that provide this guarantee\footnote{All other algorithms providing this guarantee do so by minimizing stronger notions of swap regret such as normal-form swap regret, which are conjectured to be hard to minimize (both information-theoretically and computationally, see \cite{daskalakis2024lowerboundswapregret}).}. 
\end{itemize}

Before proceeding to discuss these results (and others) in more detail, it is helpful to provide a definition and some intuition for profile swap regret. Consider a game  where one player (the ``learner'', running a learning algorithm) picks actions from a convex polytope $\learnset$, the other player (the ``optimizer'', playing strategically) picks actions from a convex polytope $\optset$, and the utility the learner receives from playing action $x \in \learnset$ against the action $y \in \optset$ is given by $u_L(x, y)$, where $u_L: \learnset \times \optset \rightarrow \Rset$ is some bilinear function of both players' actions. Consider also a hypothetical transcript of a repeated instance of this game, where the learner has played the sequence of actions $\bx = (x_1, x_2, \dots, x_T)$ and the optimizer has played the sequence of actions $\by = (y_1, y_2, \dots, y_T)$. The standard (external) regret of the learner is simply the gap between the utility $\sum_{t} u_L(x_t, y_t)$ that they received and the best possible utility $\sum_{t} u_L(x^{*}, y_t)$ they could have received in hindsight by committing to a fixed action $x^{*} \in \learnset$. 

How should we define the swap regret of the learner on this transcript? If this were a normal-form game where the learner had $m$ actions, we could define swap regret by comparing the learner's utility to the best possible utility they could obtain by applying a ``swap function'' $\pi: [m] \rightarrow [m]$ to their actions, i.e., $\sum_{t} u_L(\pi(x_t), y_t)$. Implicit in this definition is the fact that we can extend the swap function $\pi$ from the set of pure strategies ($[m]$) to the set of mixed strategies (the simplex $\Delta_{m}$). In the case of normal-form games, there is a clear way to do this: decompose each mixed action as a combination of pure strategies and apply the swap function to each of the pure strategies in this decomposition.

In general polytope games, however, there may not be a unique way to decompose mixed strategies of the learner (elements of $\learnset$) into pure strategies (vertices of $\learnset$). For example, if $\learnset = [0, 1]^2$, it is possible to decompose the mixed action $x = (1/2, 1/2)$ as either $\frac{1}{2}(0, 0) + \frac{1}{2}(1, 1)$, or as $\frac{1}{2}(0, 1) + \frac{1}{2}(1, 0)$, and these two decompositions could get sent to very different mixed strategies under an arbitrary swap function mapping pure strategies to pure strategies. This has led to a number of different definitions for swap regret in polytope games, including linear swap regret (where we restrict the swap functions to be given by linear transformations, so that they can act directly on mixed actions), polytope swap regret (where we pick the best possible decomposition for the learner in each round), and normal-form swap regret (where we require the learner to directly play a distribution over pure strategies).

Profile swap regret addresses the issue of non-unique convex decompositions in the following way. Given a transcript $(\bx, \by)$, we construct its corresponding \emph{correlated strategy profile (CSP)} $\csp$, which we define to equal the element $\csp = \frac{1}{T}\sum_{t} x_{t} \otimes y_{t} \in \learnset \otimes \optset$. Note that in normal-form games, the CSP $\csp$ captures the correlated distribution over pairs of pure strategies played by both players -- it plays a somewhat similar role here, allowing us to evaluate the average of any bilinear utility function over the course of play. Now, for any decomposition of $\csp$ into a convex combination $\csp = \sum_{k} \lambda_{k} (x_{(k)} \otimes y_{(k)})$ of product strategy profiles (rank one elements of $\learnset \otimes \optset$), we define the profile swap regret of this decomposition to equal the maximal increase in utility by swapping each $x_{(k)}$ to the best response to $y_{(k)}$. Finally, we define the overall profile swap regret of this transcript to be the minimum profile swap regret of any valid decomposition of this form.

This definition, although perhaps a little peculiar, has a number of interesting properties:

\paragraph{Utility-theoretic properties} One important consequence of the above definition is that profile swap regret can be computed entirely as a function of the CSP $\csp$ (unlike stronger regret notions like polytope swap regret and normal-form swap regret). This allows us to build off the work of \cite{paretooptimal}, who develop techniques for understanding the subset of possible CSPs an adversary can asymptotically induce against a specific learning algorithm (they call this set the \emph{asymptotic menu} of the learning algorithm). By defining profile swap regret in this way, straightforward generalizations of these menu-based techniques to polytope games suffice to establish the  properties of non-manipulability, minimality, and Pareto-optimality for no-profile-swap-regret algorithms.

\paragraph{Efficient no-profile-swap-regret algorithms} Another benefit of of the above definition is that the CSP $\csp$ is a fairly low-dimensional vector, living in a $(\dim(\learnset) \cdot \dim(\optset))$-dimensional vector space (contrast this with the amount of information required to store a distribution over the potentially exponential number of vertices of $\learnset$). This allows us to design learning algorithms that incur at most $O(\sqrt{T})$ profile swap regret by using Blackwell's approachability theorem to force the CSP of the transcript of the game to quickly approach the subset of CSPs with zero profile swap regret. 
    
Interestingly, the computational properties of profile swap regret (and these associated learning algorithms) are somewhat subtle. We prove that it is NP-hard to compute the profile swap regret incurred by a learner during a specific transcript of play $(\bx, \by)$, even in the special case of Bayesian games (Theorem~\ref{thm:hardness}). Ordinarily, this would preclude running the previous approachability-based algorithms efficiently. However, by extending a technique recently introduced in \cite{daskalakis2024efficient} -- optimization via semi-separation oracles -- we show that it is still  possible to implement a variant of the Blackwell approachability algorithm in polynomial time. 

\paragraph{Polytope swap regret and game-agnostic learning} In \cite{MMSSbayesian}, the authors introduced \emph{polytope swap regret} as a measure of swap regret in polytope games that guarantees non-manipulability when minimized, and asked whether any non-manipulable algorithm must necessarily minimize swap regret. Following this, \cite{rubinstein2024strategizing} answered this question affirmatively for the case of Bayesian games. However, this answer has a minor subtlety -- in order to construct Bayesian games where a particular high polytope swap regret algorithm is manipulable, their construction might use games where the optimizer has far more actions available to them than in the games where the learner incurs high regret. \cite{rubinstein2024strategizing} further point out that this subtlety is in some sense unavoidable -- if you fix the number of actions and types of the two players, there are learning algorithms which incur high polytope swap regret but that are not manipulable. 

We study this phenomenon in general polytope games by drawing a distinction between \emph{game-aware learning algorithms} -- algorithms that can see the sequence of actions $\by \in \optset$ the optimizer is directly playing -- and \emph{game-agnostic learning algorithms} -- algorithms that can only see the sequence of induced counterfactual rewards (i.e., the function sending an action $x \in \learnset$ to the utility $u(x, y_t)$ they would have received if they played $x$ in round $t$). While profile swap regret characterizes non-manipulability for game-aware algorithms, we show that  polytope swap regret characterizes non-manipulability for game-agnostic  algorithms (Theorem~\ref{thm:agnostic-main}), thus providing an analogue of the result of \cite{rubinstein2024strategizing} for general polytope games. 

    \paragraph{Implications for equilibrium computation} We finally turn our attention to the question of how to define and compute correlated equilibria (CE) in polytope games. We argue that there are two main motivations for studying correlated equilibria, which are often conflated due to their agreement in the case of normal-form games. The first is the idea that a correlated equilibrium is an outcome that is inducible by a third party mediator providing correlated recommendations to all players. This idea is appealing from a mechanism design point-of-view, as one can imagine directly implementing the correlated equilibrium of our choice in a game by constructing such a mediator (e.g., installing a traffic light). 

    However, in many settings of interest, there is no explicit mediator. A second motivation for studying CE is that a correlated equilibrium represents a possible outcome of repeated strategic play between rational agents. We can still relate this back to our original mediator motivation by saying that an outcome is a correlated equilibrium if every player can individually imagine a ``one-sided'' mediator protocol incentivizing this scheme, where only that player's incentive constraints need to be met (this is in contrast to the ``two-sided'' mediator protocol above, which must work for all players simultaneously).
    
    We show that, in polytope games, this first form of CE corresponds to \emph{normal-form CE} (reached by normal-form swap regret dynamics) whereas the second form of CE corresponds to \emph{profile CE} (reached by profile swap regret dynamics). We further show that there is a gap between these two notions of correlated equilibria in general polytope games, exhibiting a separation which does not appear in the case of normal-form games. Finally, we note that our efficient no-profile-swap-regret algorithms mentioned above allow us to compute a profile CE in general games without being able to optimize over the set of profile CE, echoing the results of \cite{papadimitriou2008computing} for computing correlated equilibria in succinct multiplayer games.

\subsection{Related Work}


\paragraph{Swap Regret and $\Phi$-regret} Swap regret has long been an object of interest in normal-form games \citep{foster1997calibrated}. Early efficient methods to achieve bounded swap regret were established by~\cite{blum2007external}, who proposed efficient algorithms that guarantee low internal regret. More recently, concurrent work by ~\cite{dagan2023external} and ~\cite{peng2024swap} have shown that it is possible to minimize swap regret in any online learning setting where it is possible to minimize external regret. 

\cite{GordonGreenwaldMarks2008} introduced a generalization of swap regret to convex games called $\Phi$-regret, where one competes with a family $\Phi$ of functions mapping the action set into itself. This generalization captures many other swap regret notions of interest. Restricting $\Phi$ to only contain linear functions, we obtain linear swap regret. Linear swap regret has recently been studied extensively in both Bayesian games \citep{MMSSbayesian, fujii2023bayes, dann2023pseudonorm} and extensive form games \citep{Farina2023:Polynomial, Farina2024:eah, zhang2023mediator}, with \cite{daskalakis2024efficient} providing efficient algorithms for minimizing linear swap regret in general polytope games. Recently \cite{zhang2024efficient} studied $\Phi$-regret minimization for classes of functions $\Phi$ specified by low degree polynomials. Finally, \citep{dagan2023external, peng2024swap, fishelsonfull} study the notion of \emph{full swap regret}, allowing $\Phi$ to be the set of all (non-linear) functions mapping the action set into itself. There are some other variants of $\Phi$-regret minimization studied towards the goal of computing specific variants of correlated equilibria in Bayesian or extensive-form games; we survey those below. 

\paragraph{Strategizing in games} While no-swap regret algorithms were first conceptualized via their connection to correlated equilibrium, a more recent line of work has investigated the strategic properties of these algorithms in their own right.~\cite{braverman2018selling} initiated the study of non-myopic responses to learning algorithms in the context of single buyer auctions, demonstrating that when bidders run standard learning algorithms to choose their bids, they can be fully manipulated by a seller (who can extract the full surplus of the auction, leaving the buyer with zero utility). Since then, there has been a large line of recent work focused on understanding which learning algorithms provide provable game-theoretic guarantees in settings such as auctions \citep{deng2019prior, cai2023selling, lin2024persuading, rubinstein2024strategizing}, principal-agent problems \citep{guruganesh2024contracting, lin2024persuading}, general normal-form games \citep{deng2019strategizing, brown2024learning,  haghtalab2024calibrated, camara2020mechanisms}, and Bayesian games \cite{MMSSbayesian, rubinstein2024strategizing}.

\paragraph{Correlated equilibria in polytope games} The concept of correlated equilibria originates from~\cite{aumann1974subjectivity} as a generalization of the notion of a Nash equilibrium for players who can correlate their play. Correlated equilibria also have the nice property that unlike Nash equilibria, they are computable in polynomial time \citep{papadimitriou2008computing}, at least in normal-form games. One of the main motivations for designing no-swap-regret learning algorithms is to construct decentralized learning dynamics that provably converge to correlated equilibria at fast rates (e.g, \citep{anagnostides2022near}).

On the other hand, the simplest generalization of correlated equilibria to general polytope games -- that is, to \emph{normal-form correlated equilibria (NFCE)}, formed by each extremal strategy as a pure strategy in the corresponding game -- might blow up the size of the game exponentially, and there is therefore no known efficient algorithm for computing NFCE. Moreover, there is no clear way to ``correctly'' generalize the original definition of \cite{aumann1974subjectivity} to these settings -- \cite{forges1993five} introduces ``five legitimate definitions of correlated equilibrium'' for games with sequential imperfect information. In Bayesian games, \cite{bergemann2016bayes} introduce a notion of Bayes correlated equilibrium, but did not discuss computational aspects; more recently, \cite{fujii2023bayes} studies three refinements of this notion (agent-normal-form CE, communication equilibria, and strategic-form CE), and shows how to compute communication equilibria by minimizing ``untruthful swap regret'' (a variant of linear swap regret). In extensive-form games, \cite{von2008extensive} introduce the concept of \emph{extensive-form correlated equilibria}, which can be computed in polynomial-time in the representation of the game either directly \citep{huang2008computing} or by decentralized learning dynamics \citep{farina2022simple} (minimizing ``trigger swap regret''). 

\paragraph{Blackwell approachability} The main technique we use to design efficient learning algorithms for minimizing profile swap regret is an application of the semi-separation framework of \cite{daskalakis2024efficient} to the general problem of Blackwell approachability \citep{blackwell1956analog}. \cite{abernethy2011blackwell} demonstrated a reduction from Blackwell approachability to  regret-minimization that we use in this application. The orthant-approachability form of Blackwell approachability that we introduce in Section~\ref{sec:algorithms} appears implicitly in many follow-up works that focus on improving the rates of approachability algorithms \citep{perchet2013approachability, Perchet2015, Kwon2021, dann2023pseudonorm, dann2024rate}.




\section{Model and Preliminaries}

\paragraph{Notation} Given two convex sets $\cC_1 \subseteq \Rset^{d_1}$ and $\cC_2 \subseteq \Rset^{d_2}$, we define their tensor product $\cC_1 \otimes \cC_2$ to be the subset of $\Rset^{d_1} \otimes \Rset^{d_2} \simeq \Rset^{d_1d_2}$ equal to the convex hull of all vectors of the form $c_1 \otimes c_2 = c_1 c_2^\top$ for $c_1 \in \cC_1$ and $c_2 \in \cC_2$. We write $\ball_{d}(R)$ to denote the $d$-dimensional ball of radius $R$ centered at the origin, and $\ball_{d}(x, R)$ to denote the $d$-dimensional ball of radius $R$ centered at $x$ (sometimes omitting $d$ when it is clear from context).

Unless otherwise specified, all norms $||\cdot||$ refer to the $\ell_2$ norm in the ambient space. Note that any bilinear function $f: \Rset^{d_1} \times \Rset^{d_2} \rightarrow \Rset$ corresponds to a vector $\hat{f}$ such that $f(x, y) = \langle \hat{f}, x \otimes y\rangle$; we define the various norms of $f$ (e.g. $||f||_{1}$, $||f||_{2}$, $||f||_{\infty}$) to equal those of $\hat{f}$. 

Selected proofs are omitted and deferred to Appendix~\ref{sec:omitted}.

\subsection{Polytope Games}

We begin by introducing the notion of an \emph{polytope game}: a two-player game where the action sets of both players are convex polytopes\footnote{Almost all the results should extend to the slightly more general setting of arbitrary bounded convex sets. We focus on the specific setting of polytopes as this captures essentially all the primary settings of interest (e.g. Bayesian games and extensive-form games).} and where the payoffs of both players are provided by bilinear functions over these two sets. Formally, a polytope game $G$ is a game between two players; we call these two players the \emph{learner} and the \emph{optimizer}. The learner selects an action $x$ from the $d_L$-dimensional bounded convex polytope $\learnset \subset \ball_{d_L}(1)$ and the optimizer selects an action $y$ from the bounded convex polytope $\optset \subset \ball_{d_O}(1)$ (note that $\learnset$ and $\optset$ can have different dimensions). After doing so, the learner receives utility $u_L(x, y)$ and the optimizer receives utility $u_O(x, y)$, where $u_L$ and $u_O$ are both bounded bilinear functions satisfying $||u_L||_{\infty}, ||u_O||_{\infty} \leq 1$. These two action sets ($\learnset$ and $\optset$) and payoff functions ($u_L$ and $u_O$) define the game $G$.

We briefly note here that polytope games (for specific choices of $\learnset$ and $\optset$) capture a variety of different strategic settings. For example:

\begin{itemize}
    \item \textbf{Normal-form games}: When $\learnset = \Delta_m$ and $\optset = \Delta_n$, this captures the class of normal-form games where the learner has $m$ available actions and the optimizer has $n$ available actions.
    \item \textbf{Bayesian games}: When $\learnset = (\Delta_{m})^{c_L}$ and $\optset = (\Delta_{n})^{c_O}$, this captures the class of Bayesian games where again the learner and optimizer have $m$ and $n$ available actions respectively, but in addition the learner is one of $c_L$ types and the optimizer is one of $c_O$ types. Here we should interpret the set $(\Delta_{m})^{c_L}$ as representing all functions mapping a learner's type (an element of $[c_L]$) to a learner's mixed action (an element of $\Delta_m$); we can interpret the set $(\Delta_{n})^{c_O}$ similarly. 
    \item \textbf{Extensive-form games}: A two-player extensive form game can be written as a polytope game by letting $\learnset$ and $\optset$ be the sequence-form polytopes for the two players. See e.g. \cite{von1996efficient} for details.
\end{itemize}

Generally, we will be interested in polytope games that are played repeatedly for $T$ rounds (and even more generally, the limiting behavior of such games as $T$ approaches infinity). In a repeated setting, we will let $x_t \in \learnset$ denote the action taken by the learner at round $t$, and let $y_t \in \optset$ denote the action taken by the optimizer at round $t$. Given a transcript of a repeated game where the learner has played the sequence of actions $x_1, x_2, \dots, x_T$ and the optimizer has played the sequence of actions $y_1, y_2, \dots, y_T$, we say that the \emph{correlated strategy profile (CSP)} $\csp$ corresponding to this transcript is given by

\begin{equation}\label{eq:csp-def}
\csp = \frac{1}{T}(x_1 \otimes y_1 + x_2 \otimes y_2 + \dots + x_T \otimes y_T) \in \learnset \otimes \optset.
\end{equation}

\noindent
By construction, the CSP $\csp$ provides sufficient information to evaluate the average value of any bilinear function $f(x_t, y_t)$ over the transcript of play. Under the minor assumption\footnote{This assumption is true for the examples mentioned above, and can always be made to be true by augmenting $\learnset$ and $\optset$ with an additional dummy coordinate fixed to equal $1$.} that $\learnset$ and $\optset$ are contained within proper affine subspaces of $\Rset^{d_L}$ and $\Rset^{d_O}$ respectively, the CSP $\csp$ is also sufficient to evaluate the average of any bi-affine function $f(x_t, y_t)$ over this transcript. We will therefore assume that $\learnset$ and $\optset$ satisfy this assumption throughout the rest of the paper, and write $f(\csp) = \frac{1}{T}\sum_{t}f(x_t, y_t)$ for any such bi-affine function $f$ (e.g, $u_L(\csp)$ is the average learner utility over this transcript of play).

Given a polytope game, we will write $\BR_{L}(x)$ to denote the set of the learner's best responses to the optimizer's action $y$, i.e., $\BR_{L}(y) = \{x \in \learnset \,\mid\, u_L(x, y) = \max_{x^{*} \in \learnset} u_L(x^{*}, y)\}$. We say an action $x \in \learnset$ for the learner is \emph{strictly dominated} if it is not the best response to any action, i.e., there does not exist an $y \in \optset$ such that $x \in \BR_{L}(y)$. We say an action $x \in \learnset$ for the learner is \emph{weakly dominated} if it is \emph{not} strictly dominated but it is impossible for the optimizer to uniquely incentivize $x$, i.e., there does not exist an $y \in \optset$ such that $\BR_{L}(y) = \{x\}$. We say a polytope game $G$ is \emph{non-degenerate} if none of the vertices of $\learnset$ (the learner's extremal actions) are weakly dominated. 

\subsection{Learning Algorithms and Regret}

Thus far, we have made no distinction between the role of the optimizer and the learner. The difference, of course, is that the learner will play this repeated game by running a learning algorithm. 

A \emph{learning algorithm} $\cA$ for the game $G$ is a family $\{\cA^{T}\}_{T \in \mathbb{N}}$ of horizon-dependent learning algorithms. A horizon-dependent learning algorithm $\cA^{T}$ for time horizon $T$ is a collection of $T$ functions $\cA_1^{T}, \cA_2^{T}, \dots, \cA_T^T$, where $\cA_{t}^{T}$ describes the learner's play at time $t$ as a function of the optimizer's play up until round $t-1$, i.e., $\cA_{t}^{T}(y_1, y_2, \dots, y_{t-1}) = x_t$. Note that as written, a learning algorithm is specific to the given game $G$, and a learning algorithm for one game $G$ cannot necessarily easily be applied to a separate game $G'$. The important caveat to this is that a learning algorithm does not depend on the optimizer's utility function $u_O$, and we expect a robust learning algorithm $\cA$ for a game $G$ to perform well against several different opponents with different choices of $u_O$. 

We can measure the performance of a learning algorithm via some version of \emph{regret}. For example, the \emph{external regret} of a learning algorithm is the gap between its utility and the counterfactual utility it would have received if it played the best fixed action in hindsight. Formally, if the learner has played the sequence of actions $\bx = (x_1, x_2, \dots, x_T)$ and the optimizer has played the sequence of actions $\by = (y_1, y_2, \dots, y_T)$, the external regret of the learner is given by

\begin{equation*}
\Reg(\bx, \by) = \max_{x^{*} \in \cX} \sum_{t=1}^{T} u_L(x^{*}, y_t) - \sum_{t=1}^{T} u_L(x_t, y_t).
\end{equation*}

Note that since $u_L$ is bilinear, we can also compute the external regret directly from the CSP $\csp$ corresponding to this transcript of play. In particular, we can alternatively write

\begin{equation*}
\Reg(\csp) = T \cdot \left(\max_{x^{*} \in \cX}u_L(x^* \otimes\mathrm{proj}_{\cY}(\csp)  )  - u_L(\csp) \right),
\end{equation*}

\noindent
where $\mathrm{proj}_{\cY}(\csp) = \frac{1}{T}(y_1 + y_2 + \dots + y_T)$~\footnote{This operation is well defined -- we can extract the marginal action of the optimizer from the CSP $\phi$ 
 since it an average of a bi-affine function $g(x_t,y_t) = y_t$ applied to each entry in the transcript.} is the projection of $\csp$ onto the optimizer's action space (i.e., the average action taken by the optimizer over the course of the game).

In this paper we will primarily care about forms of regret stronger than external regret, namely variants of \emph{swap regret}. For the case of normal-form games (when the learner's action set $\learnset = \Delta_m$ is just the simplex over $m$ pure actions), the swap regret of a transcript of play can be defined via

\begin{equation}\label{eq:swap-reg}
\Swap(\bx, \by) = \max_{\pi^{*}:[m]\rightarrow [m]} \sum_{t=1}^{T} u_L(\pi^*(x_t), y_t) - \sum_{t=1}^{T} u_L(x_t, y_t).
\end{equation}

Intuitively, this should be thought of as the gap between the utility received by the learner and the maximum utility the learner could have received if they applied a specific ``swap function'' $\pi^*$ to their sequence of play: a function which transforms each pure strategy $i \in [m]$ of a learner to a new pure strategy $\pi(i) \in [m]$. Implicit in this definition is the fact that although $\pi^*$ is a function on \emph{pure strategies} (elements of $[m]$), we can uniquely extend $\pi^*$ to act on \emph{mixed strategies} (elements $x \in \Delta_m$) via $\pi^*(x)_{i} = \sum_{j \mid \pi(j) = i} x_j$. That is, the weight of action $i$ in $\pi^*(x)$ is equal to the total weight of actions $j$ in $x$ that map to $i$ under $\pi^*$. 



When the learner's action set $\learnset$ is an arbitrary polytope (and not a simplex), it is not clear how to perform this extension from swap functions on ``pure strategies'' (vertices of $\learnset$) to swap functions on ``mixed strategies'' (points within $\learnset$). For this reason, it is not clear how to directly generalize the notion of swap regret to polytope games, and a couple different plausible definitions have previously been proposed. We present three of these definitions below, in increasing order of strength\footnote{Meaning that linear swap regret will always be the smallest (and easiest to minimize) and normal-form swap regret will always be the largest (and hardest to minimize).}: \emph{linear swap regret}, \emph{polytope swap regret}, and \emph{normal-form swap regret}.

\paragraph{Linear swap regret}

In linear swap regret, we constrain all of our swap functions to be linear transformations, which can then be applied directly to the mixed strategies played by the learner. Specifically, let $\Psi$ be the set of all affine linear transformations that map the set $\cX$ to itself. The \emph{linear swap regret} of a transcript with learner actions $\bx = (x_1, x_2, \dots, x_T)$ and optimizer actions $\by = (y_1, y_2, \dots, y_T)$ is given by

\begin{equation*}
\LinSwap(\bx, \by) = \max_{\psi \in \Psi} \sum_{t=1}^{T} u_L(\psi(x_t), y_t) - \sum_{t=1}^{T} u_L(x_t, y_t).
\end{equation*}

As was the case with external regret (and swap regret over the simplex), we can write linear swap regret as a function of the CSP $\csp$ corresponding to this transcript of play via

\begin{equation*}
    \LinSwap(\csp) = T \cdot \left(\max_{\psi \in \Psi} u_L(\psi(\csp)) - u_L(\csp)\right),
\end{equation*}

\noindent
where we extend $\psi$ to act on CSPs $\csp$ via $\psi(x \otimes y) = \psi(x) \otimes y$.

\paragraph{Polytope swap regret}

Another way to adapt the definition of swap regret to polytope games is to still keep swap functions that act  (possibly non-linearly) on the set of pure strategies, but choose the decomposition of the learner's mixed actions into pure actions in the way that is optimal for the learner. This is the approach taken by \cite{MMSSbayesian} in the definition of polytope swap regret.

Formally, we define the \emph{polytope swap regret} $\PolySwap(\bx, \by)$ of a transcript with learner actions $\bx = (x_1, x_2, \dots, x_T)$ and optimizer actions $\by = (y_1, y_2, \dots, y_T)$ as follows. 

\begin{enumerate}
    \item First, decompose each of the learner's actions $x_t$ into a convex combination of the extreme points $\learnvert$ of $\learnset$. We will denote this decomposition by $\xV_t \in \Delta(\learnvert)$. Note that there may be many ways to do this decomposition: we are free to choose any of them and will eventually choose the decompositions that \emph{minimize} our eventual regret.

    \item For any swap function $\pi: \learnvert \rightarrow \learnvert$, we will let $\pi(\xV_t) \in \Delta(\learnvert)$ denote the resulting distribution over the extreme points of $\learnset$ after applying $\pi$ (i.e., in the same manner as in the definition of swap regret), and $\overline{\pi}(\xV_t)$ be the element of $\learnset$ formed by taking the average action in $\pi(\xV_t)$. 

    \item Finally, we define

    \begin{equation*}
    \PolySwap(\bx, \by) = \min_{\bxV} \max_{\pi: \learnvert \rightarrow \learnvert} \left(\sum_{t=1}^{T} u_L(\overline{\pi}(\xV_t), y_t) - u_L(x_t, y_t) \right).
    \end{equation*}

    That is, after the learner picks their preferred decomposition of actions into $\xV_t$, the adversary picks the swap function $\pi$ that maximizes the resulting regret from comparing the original sequence of actions $x_t$ to the transformed sequence of actions $\overline{\pi}(\xV_t)$.
\end{enumerate}


\paragraph{Normal-form swap regret}

Finally, we can attempt to directly use the original definition of swap regret by ``expanding'' any polytope game to a normal-form game where pure strategies for the optimizer and learner correspond to extreme points of $\learnset$ and $\optset$ (equivalently, forcing the learner to specify their own decomposition of their mixed strategies into pure strategies each round). In particular, we define the \emph{vertex game} corresponding to a polytope game $G$ to be the normal-form game where the learner plays mixtures of actions in $\learnvert$ and the optimizer plays mixtures of actions in $\optvert$ (i.e., with $\learnset' = \Delta(\learnvert)$ and $\optset' = \Delta(\optvert)$). If the learner plays a sequence of vertex game actions $\bxV = (\xV_1, \xV_2, \dots, \xV_T)$ and the optimizer plays a sequence of vertex game actions $\byV = (\yV_1, \yV_2, \dots, \yV_T)$, the \emph{normal-form swap regret} of this transcript of play is given by

\begin{equation}\label{eq:norm-swap-reg}
\NormSwap(\bxV, \byV) = \max_{\pi^{*}:\learnvert\rightarrow \learnvert} \sum_{t=1}^{T} u_L(\pi^*(\xV_t), \yV_t) - \sum_{t=1}^{T} u_L(\xV_t, \yV_t).
\end{equation}

\noindent
Note that this is simply the original definition of swap regret \eqref{eq:swap-reg} as applied to the vertex game.

Technically, this definition of swap regret is of a different flavor from the previous definitions, in that there is no clear way to take a transcript of play $\bx = (x_1, \dots, x_T)$ and $\by = (y_1, \dots, y_T)$ of the original polytope game and evaluate the normal-form swap regret of this transcript. However, note that we can always go the other direction -- given a transcript of play $(\bxV, \byV)$ of the vertex game, we can always construct a corresponding transcript of play $(\bx, \by)$ for the polytope game (e.g., by letting $y_t = \E[\yV_t] \in \optset$) and evaluate the other measures of swap regret on this transcript.

\paragraph{Comparisons between regret definitions}

We conclude this section with a couple of comparisons between regret notions. First, we note that for the case of normal-form games, all three of these notions reduce to the standard notion of swap regret (as we would expect).

\begin{theorem}\label{thm:nfg-swap-notions}
Fix a polytope game $G$ with $\learnset = \Delta_m$. Then for any transcript of play $\bx = (x_1, x_2, \dots, x_T)$ and $\by = (y_1, y_2, \dots, y_T)$ in $G$, we have that

$$\Swap(\bx, \by) = \LinSwap(\bx, \by) = \PolySwap(\bx, \by) = \NormSwap(\bx, \by).$$

\noindent
(Note that for normal-form games, the vertex game is identical to the original game, and therefore the quantity $\NormSwap(\bx, \by)$ is well-defined).
\end{theorem}

The following theorem shows that the above three forms of swap regret are ordered as originally described. Moreover, this ordering is ``strict'' -- minimizing one of these regret notions implies nothing about the larger regret notions. 

\begin{theorem}\label{thm:swap-notions-ordering}
Let $\bx = (x_1, x_2, \dots, x_T)$ and $\by = (y_1, y_2, \dots, y_T)$ be a transcript of play for a given polytope game. Let $\bxV = (\xV_1, \xV_2, \dots, \xV_T)$ and $\byV = (\yV_1, \yV_2, \dots, \yV_T)$ be a transcript of play for the corresponding vertex game such that $x_t = \E[\xV_t]$ and $y_t = \E[\yV_t]$ for each $t \in [T]$. Then:

$$\LinSwap(\bx, \by) \leq \PolySwap(\bx, \by) \leq \NormSwap(\bxV, \byV).$$

Moreover, all these inequalities are asymptotically strict in the following sense: for each two neighboring definitions of swap regret, there exists a family of transcripts for which the larger swap regret grows as $\Omega(T)$ but the smaller swap regret is zero.
\end{theorem}

We will often be interested in the worst-case regret incurred by some algorithm. For a horizon-dependent learning algorithm $\cA^{T}$, we let $\Reg(\cA^{T})$ denote the maximum value of $\Reg(\bx, \by)$ over all transcripts of length $T$ obtainable by playing against $\cA^{T}$. For a general learning algorithm $\cA$, we will let $\Reg(\cA)$ denote the function mapping $T$ to $\Reg(\cA^{T})$ (so e.g., we may have $\Reg(\cA) = O(\sqrt{T})$). We extend this notation to other notions of regret (e.g., $\Swap$, $\LinSwap$, $\PolySwap$, $\NormSwap$) in the obvious way.

\subsection{Manipulability and Menus}\label{sec:manipulability}

Given the above range of possible definitions for swap regret in polytope games, which ones should we target when learning in games? Of course, we can always try to minimize the strongest form of swap regret above (normal-form swap regret), but this might come with trade-offs in the form of worse regret bounds and increased algorithmic complexity. 

Instead, a more principled approach is to understand why we might want to minimize swap regret in normal-form games in the first place, and then pick the variant of swap regret that gives us comparable guarantees for polytope games. In particular, swap-regret minimization has a number of nice game-theoretic consequences in the form of convergence to certain classes of equilibria and robustness of the learner to certain dynamic manipulations of the optimizer. Informally, we can state some of these consequences as follows:

\begin{enumerate}
\item \textbf{(Non-manipulability)} No-swap-regret learning algorithms are ``non-manipulable'', in the sense that if a learner is running a no-swap-regret learning algorithm, the optimizer can do nothing asymptotically better than playing a fixed (possibly mixed) action every round. Moreover, every non-manipulable no-regret learning algorithm must also be no-swap-regret.

\item \textbf{(Minimality)} No-swap-regret learning algorithms form a ``minimal core'' of all no-regret learning algorithms: any CSP $\csp$ that an optimizer can implement against a no-swap-regret learning algorithm can be implemented against any no-regret algorithm.

\item \textbf{(Pareto-optimality)} No-swap-regret learning algorithms are Pareto-optimal: there is no learning algorithm that performs asymptotically better than a no-swap-regret learning algorithm against every possible optimizer. Notably, some classic no-regret learning algorithms like Follow-The-Regularized-Leader are \emph{not} Pareto-optimal in this sense.

\item \textbf{(Correlated equilibria)} When both players in the game run no-swap-regret learning algorithms, the time-averaged CSP converges to the set of correlated equilibria of the game.
\end{enumerate}

We will table the discussion of equilibria until Section \ref{sec:equilibria} -- it is complicated by the fact that it also is not exactly clear what the exact definition of correlated equilibria should be for polytope games. Instead, we will focus on the first three points for now, which all have clear utility-theoretic interpretations that we can extend to general polytope games. To more formally define them, it is useful to introduce the concept of menus.

Given a horizon-dependent learning algorithm $\cA^{T}$, we define the \emph{menu} $\cM(\cA^{T})$ of this algorithm to be the convex hull of all CSPs that an optimizer can implement against $\cA^{T}$. That is, $\cM(\cA^{T}) \subseteq \learnset \otimes \optset$ is the convex hull of all points of the form $\frac{1}{T}\sum_{t=1}^{T} x_t \otimes y_t$ where $x_t = \cA_t^{T}(y_1, y_2, \dots, y_{t-1})$, and $y_1, y_2, \dots, y_T$ is an arbitrary sequence of optimizer actions in $\optset$. For a general learning algorithm $\cA$, we define the \emph{asymptotic menu} $\cM(\cA)$ to be the limit of the sequence of menus $\cM(\cA^{1}), \cM(\cA^{2}), \dots$ in the Hausdorff norm\footnote{Throughout this paper, for simplicity, we will assume that this sequence of convex sets always converges. In general, one can always take some subsequence of this sequence that converges -- see Appendix D of \cite{paretooptimal} for details.}.

There is a simple test (coming from Blackwell approachability) for whether a convex set $\cM$ is the asymptotic menu of some learning algorithm: it suffices that $\cM$ contain some CSP of the form $x \otimes y$ for each $y \in \optset$. 
\begin{theorem}\label{thm:menu_char}
    A convex set $\cM$ is the asymptotic menu of some learning algorithm iff $\cM$ contains some CSP of the form $x \otimes y$ for each $y \in \optset$. 
\end{theorem}

Given a polytope game $G$ (with a specific optimizer payoff $u_O$), we define the \emph{Stackelberg value} $\Stack(G, u_O)$ to be the optimal payoff the optimizer can receive if they commit to playing a fixed strategy $y \in \optset$ and the learner best responds (breaking ties in favor of the optimizer). That is,

\begin{equation*}
\Stack(G, u_O) = \max_{y \in \optset} \max_{x \in \BR_{L}(y)} u_O(x, y).
\end{equation*}

\noindent
We make the dependence on $u_O$ explicit here because we will often want to consider the effect of changing $u_O$ while keeping all other parameters of the game ($\learnset$, $\optset$, and $u_L$) the same -- this captures the strategic problem of facing an optimizer with unknown rewards. 

Against any no-(external)-regret learning algorithm $\cA$, an optimizer can asymptotically achieve $\Stack(G, u_O)$ utility per round by simply playing their optimal Stackelberg action every round. Intuitively, an algorithm of the learner is non-manipulable if the optimizer cannot significantly increase their utility beyond this Stackelberg value by playing a strategy that changes over time. 

It is convenient to phrase this concept of non-manipulability in the language of menus. For any menu $\cM$, define $V_O(\cM, u_O) = \max_{\csp \in \cM} u_O(\csp)$ to be the maximum optimizer utility of any CSP $\csp$ in $\cM$. If $\cM$ is the menu of an algorithm $\cA$, note that this is just the maximum utility an optimizer can achieve by playing against algorithm $\cA$. We therefore say the menu $\cM$ is \emph{non-manipulable} if, for any optimizer payoff $u_O$,

\begin{equation*}
V_O(\cM, u_O) \leq \Stack(G, u_O).
\end{equation*}

The first point can be now rephrased as follows\footnote{The remaining theorems in this section all follow as consequences of results in \cite{deng2019strategizing}, \cite{MMSSbayesian}, and \cite{paretooptimal}. Since they are also special cases of theorems we prove later in this paper for polytope games (Theorems \ref{thm:poly_nonmanip}, \ref{thm:poly_minimal}, and \ref{thm:poly_pareto}), we omit their proofs.}.

\begin{theorem}\label{thm:nfg_nonmanip}
Fix a normal-form game $G$ (i.e., where $\learnset = \Delta_{m}$ and $\optset = \Delta_{n}$). Let $\cA$ be a no-swap-regret algorithm for $G$. Then the asymptotic menu $\cM(\cA)$ is non-manipulable. Conversely, if the asymptotic menu $\cM(\cA)$ of an algorithm $\cA$ is non-manipulable, then the algorithm $\cA$ must additionally be a no-swap-regret algorithm. 
\end{theorem}

We can also provide a version of Theorem~\ref{thm:nfg_nonmanip} quantifying the degree of manipulability of a finite-horizon algorithm. To this end, we define a menu $\cM$ to be \textit{$\alpha$-non-manipulable} if, for any payoff $u_O$, $V_O(\cM, u_O) \leq \Stack(G, u_O) + \alpha$.


Interestingly, the relevant swap-theoretic quantity for tightly characterizing the degree of manipulability of an algorithm is not swap regret itself\footnote{To see why, note that the swap regret of an algorithm scales with the learner's utility function $u_L$, but the degree of manipulability is independent of the scale of $u_L$.}, but the \emph{swap regret distance}, which captures the maximum distance between the menu of this algorithm and the no-swap-regret menu. Formally, for any CSP $\csp$, we define $\SwapDist(\csp)$ to equal the maximum $\ell_2$-distance $\mathrm{dist}(\csp, \cM_{NSR})$ from $\csp$ to the no-swap-regret menu $\cM_{NSR}$. Then, for any algorithm $\cA$, we define $\SwapDist(\cA) = \SwapDist(\cM(\cA)) = \max_{\csp \in \cM(\cA)} \SwapDist(\csp)$. 

\begin{theorem}\label{thm:nfg_nonmanip_quant}
Fix a normal-form game $G$ (i.e., where $\learnset = \Delta_{m}$ and $\optset = \Delta_{n}$). If $\cA$ is a learning algorithm for $G$ with the guarantee that $\SwapDist(\cA^{T}) \leq R(T)$, then $\cM(\cA^{T})$ is $\sqrt{mn} \cdot (R(T)/T)$-non-manipulable. Conversely, if $\cM(\cA^{T})$ is $(R(T)/T)$-non-manipulable, then the algorithm $\cA$ must satisfy $\SwapDist(\cA^{T}) \leq R(T)$. 
\end{theorem}


Moving onto the second point, we say that a menu $\cM$ is \emph{minimal} if: i. it is the asymptotic menu of some learning algorithm $\cA$, and ii. there is no strict sub-menu $\cM' \subset \cM$ which is the asymptotic menu of a learning algorithm. We can now rephrase the second point as follows. 

\begin{theorem}\label{thm:nfg_minimal}
Fix a normal-form game $G$ (i.e., where $\learnset = \Delta_{m}$ and $\optset = \Delta_{n}$). All no-swap-regret algorithms $\cA$ for $G$ share the same asymptotic menu $\cM(\cA) = \cM_{NSR}$. Moreover, if $\cA$ is a no-regret algorithm, then $\cM_{NSR} \subseteq \cM(\cA)$.
\end{theorem}

Finally, the result on non-manipulability (Theorem~\ref{thm:nfg_nonmanip}) concerns the utility the \textit{optimizer} can obtain when playing against a no-swap-regret learning algorithm. But we may also wonder when it is in the interest of the learner to play such a learning algorithm. To this end, define 

$$V_L(\cM, u_O) = \max \left\{u_L(\csp) \mid \csp \in \cM, u_O(\csp) = V_O(\cM, u_O)\right\}$$ 

\noindent
to be the utility the learner receives when an optimizer chooses their favorite CSP in $\cM$ (breaking ties in favor of the learner).

\begin{theorem}\label{thm:nfg_pareto}
Fix a normal-form game $G$ (i.e., where $\learnset = \Delta_{m}$ and $\optset = \Delta_{n}$). Let $\cA$ be a learning algorithm for $G$ that incurs $\Omega(T)$ swap regret in the worst-case. Then there exists an optimizer utility $u_O$ such that $V_{L}(\cM_{NSR}, u_O) > V_{L}(\cM(\cA), u_O)$; i.e., there exists an optimizer against whom it is strictly better to play any no-swap-regret algorithm than $\cA$.

\end{theorem}

Theorem \ref{thm:nfg_pareto} can also be thought of as saying that the no-swap-regret menu is \emph{Pareto-optimal}: there is no other learning algorithm which is asymptotically at least as good as swap regret minimization against every single possible optimizer (while strictly better for at least one optimizer). Note that unlike the other two points, this is not a tight characterization of no-swap-regret algorithms -- there exist other Pareto-optimal learning algorithms that incur $\Omega(T)$ swap regret in the worst-case \cite{paretooptimal}. 



\section{Profile Swap Regret}

\subsection{Defining Profile Swap Regret}

Like polytope swap regret, profile swap regret involves performing a convex decomposition of the sequence of play of the learner into elements that we can then compute the swap regret of and aggregate. However, unlike polytope swap regret (which decomposes the strategy of the learner in each round of the game), in profile swap regret we will do this decomposition on the average CSP of play. In particular, profile swap regret can be computed as a function of just the resulting CSP $\csp$ of play -- note that this is a property shared with linear swap regret, but that polytope swap regret and normal-form swap regret do not possess.

To define the \emph{profile swap regret} $\CorrSwap(\csp)$ of a CSP $\csp$, we perform the following steps\footnote{In Appendix~\ref{app:alternate-formulation}, we present an alternate formulation of profile swap regret in a way that is more directly comparable with the definition of polytope swap regret above.}:

\begin{enumerate}
    \item First, decompose $\csp$ into a convex combination of independent (product) strategy profiles

    \begin{equation}\label{eq:csp-decomp}
        \csp = \sum_{k=1}^{K} \lambda_{k} (x_{(k)} \otimes y_{(k)}).
    \end{equation}

    for some choice of $x_{(k)} \in \learnset$, $y_{(k)} \in \optset$, and $\lambda_k \geq 0$ with $\sum_{k} \lambda_k = 1$.  As with polytope swap regret, there are likely multiple ways to do this decomposition; we will eventually choose the decomposition that minimizes our eventual regret. Also note that we can freely choose the number of parts $K$ in this decomposition.

    \item For any $x \in \learnset$ and $y \in \optset$, let $\Reg(x, y) = \max_{x^* \in \learnset} u_L(x^{*}, y) - u_L(x, y)$ (this can be thought of as the ``instantaneous'' regret from playing $x$ in response to $y$). We define

    \begin{equation}
    \label{eq:profile_swap_regret}
    \CorrSwap(\csp) = T \cdot \left(\min_{x_{(k)}, y_{(k)}, \lambda_k} \sum_{k=1}^{K} \lambda_{k}\Reg(x_{(k)}, y_{(k)})\right),
    \end{equation}

    \noindent
    where, as mentioned above, this minimum is over all valid convex decompositions of $\csp$ into the form in \eqref{eq:csp-decomp}.
\end{enumerate}

We will eventually demonstrate that profile swap regret is the analogue of swap regret in polytope games that preserves the game-theoretic properties mentioned in Section~\ref{sec:manipulability}. Before we do so, we first remark that profile swap regret is bounded between linear swap regret and polytope swap regret.

\begin{theorem}\label{thm:comp-prof-swap}
Let $\bx = (x_1, x_2, \dots, x_T)$ and $\by = (y_1, y_2, \dots, y_T)$ be a transcript of play for a given polytope game. Then $\LinSwap(\bx, \by) \leq \CorrSwap(\bx, \by) \leq \PolySwap(\bx, \by)$.

Moreover, both of these inequalities are asymptotically strict in the following sense: for each inequality, there exists a family of transcripts for which the larger swap regret grows as $\Omega(T)$ but the smaller swap regret is zero. 
\end{theorem}

Similar to the no-swap-regret menu, we will define the \emph{no-profile-swap-regret menu} $\cM_{NPSR}$ to be the set of CSPs $\csp$ satisfying $\CorrSwap(\csp) = 0$. Likewise, just as $\SwapDist$ is the relevant quantity in Theorem~\ref{thm:nfg_nonmanip_quant}, in many of our more quantitative results it will be more convenient to work with the following ``distance'' variant of $\CorrSwap$, where we define the \emph{profile swap distance} $\CorrDist(\csp) = \mathrm{dist}(\csp, \cM_{NPSR})$ to equal the minimal Euclidean distance from $\csp$ to the no-profile-swap-regret menu $\cM_{NPSR}$. 



\subsection{Game-Theoretic Properties of Profile Swap Regret}

We will now show that the three game-theoretic properties possessed by no-swap-regret algorithms in normal-form games -- non-manipulability, minimality, and Pareto-optimality -- hold for no-profile-swap-regret algorithms in general polytope games. 

Our main tool to prove these results will be the same menu-based techniques used to prove Theorems \ref{thm:nfg_nonmanip}, \ref{thm:nfg_minimal}, and \ref{thm:nfg_pareto}, albeit extended from the standard normal-form setting to the setting of polytope games. Underlying most of these results will be the following characterization of the no-profile-swap-regret menu as the convex hull of all ``best-response CSPs'' for the learner. 

\begin{lemma}\label{lem:menu_char}
The no-profile-swap-regret menu $\cM_{NPSR}$ is the convex hull of all CSPs of the form $x \otimes y$ where $y \in \optset$ and $x \in \BR_{L}(y)$. In particular, $\Stack(G, u_O) = V_{O}(\cM_{NPSR}, u_{O})$.
\end{lemma}
\begin{proof}
\sloppy{This follows almost directly from the definition of profile swap regret. Note that if $\CorrSwap(\csp) = 0$, then this implies that we can write it in the form $\csp = \sum \lambda_{k} (x'_k \otimes y'_k)$ where each pair $(x'_k, y'_k)$ satisfies $\Reg(x'_k, y'_k) = 0$, and therefore that $x'_k \in \BR_{L}(y'_k)$ (and so we have expressed $\csp$ as a convex combination of such points). Conversely, if $x \in \BR_{L}(y)$, then $\CorrSwap(x \otimes y) = 0$ (since $\Reg(x, y) = 0$).}
\end{proof}

\subsubsection{Non-manipulability}

We begin by establishing the analogue of Theorem~\ref{thm:nfg_nonmanip} for polytope games: that for learning algorithms in polytope games, the property of having no-profile-swap-regret is equivalent to the property of being non-manipulable by a dynamic optimizer. At the same time, we establish the analogue of Theorem~\ref{thm:nfg_nonmanip_quant}, providing quantitative bounds on the degree of manipulation in terms of profile swap distance (importantly, this is the quantity we will minimize in Section \ref{sec:algorithms}, when we turn our attention to designing no-profile-swap-regret learning algorithms). 

\begin{theorem}\label{thm:poly_nonmanip}
Fix a polytope game $G$. Let $\cA$ be a no-profile-swap-regret algorithm for $G$. Then the asymptotic menu $\cM(\cA)$ is non-manipulable. Conversely, if the asymptotic menu $\cM(\cA)$ of an algorithm $\cA$ is non-manipulable, then the algorithm $\cA$ must be a no-profile-swap-regret algorithm. 

Moreover, if $\cA$ has the guarantee that $\CorrDist(\cA^{T}) \leq R(T)$, then $\cM(\cA^{T})$ is $\sqrt{d_Ld_O} (R(T)/T)$-non-manipulable. Conversely, if $\cM(\cA^{T})$ is $R(T)/T$-non-manipulable, then $\CorrDist(\cA^{T}) \leq R(T)$.
\end{theorem}

\begin{proof}
We will first show that the no-profile-swap-regret menu $\cM_{NPSR}$ is non-manipulable (therefore implying that any no-profile-swap-regret algorithm is non-manipulable). It suffices to show that, for any optimizer payoff $u_O$ and any CSP $\csp \in \cM_{NPSR}$, $u_O(\csp) \leq \Stack(G, u_O)$. But since $\Stack(G, u_O) = V_{O}(\cM_{NPSR}, u_O)$ (by Lemma~\ref{lem:menu_char}), this is immediately true. 

Conversely, if $\cA$ is not a no-profile-swap-regret algorithm, then $\cM(\cA)$ must contain a CSP $\csp \not\in \cM_{NPSR}$. By the separating hyperplane theorem, there must exist a linear payoff function $u_O$ such that $u_O(\csp) > 0$ but $u_O(\csp') \leq 0$ for all $\csp' \in \cM_{NPSR}$. But if $u_O(\csp') \leq 0$ for all $\csp' \in \cM_{NPSR}$, then we must also have $\Stack(G, u_O) \leq 0$ (by Lemma~\ref{lem:menu_char}). Therefore $V_{O}(\cM(\cA), u_O) \geq u_{O}(\csp) > \Stack(G, u_O)$, and $\cM(\cA)$ is therefore manipulable. 

\sloppy{Finally, we can quantify the degree of manipulation in the above argument as follows. If $\CorrDist(\cA^{T}) \leq R(T)$, then for any $\csp \in \cM(\cA^{T})$, there exists a $\csp' \in \cM_{NPSR}$ such that $||\csp - \csp'||_{2} \leq R(T)/T$. In particular, for any optimizer payoff $u_O$, we have that $u_{O}(\csp) - \Stack(G, u_O) \leq u_{O}(\csp) - u_{O}(\csp') \leq ||u_O||_{2}\cdot ||\csp - \csp'||_{2} \leq \sqrt{d_Ld_O}(R(T)/T)$ (where the last inequality follows from the fact that $||u_O||_{2} \leq \sqrt{d_Ld_O}$ if $||u_{O}||_{\infty} \leq 1$). }

Likewise, if $\CorrDist(\cA^{T}) \geq R(T)$, then there is a $\csp \in \cM$ such that $||\csp - \csp'||_{2} \geq R(T)/T$ for any $\csp' \in \cM_{NPSR}$. By the separating hyperplane theorem there then exists a $u_O$ with unit $||u_{O}||_{2} = 1$ such that $u_{O}(\csp) - \max_{\csp' \in \cM_{NPSR}}u_O(\csp') \geq R(T)/T$. Since $\Stack(G, u_O) = \max_{\csp' \in \cM_{NPSR}} u_O(\csp')$, this implies that the menu $\cM$ is at least $(R(T)/T)$-manipulable.
\end{proof}

We briefly remark that the $\sqrt{d_Ld_O}$ gap between the two bounds in Theorem~\ref{thm:poly_nonmanip} is an artifact of the specific boundedness constraints we placed on the optimizer's utility $u_O$ (namely, that $||u_O||_{\infty} \leq 1$). If we instead imposed the constraint $||u_O||_2 \leq 1$, then $\CorrDist$ would be an exact measure of non-manipulability.

\subsubsection{Minimality}

Next, we establish the analogue of Theorem \ref{thm:nfg_minimal}, proving that the no-profile-swap-regret menu is \emph{minimal}: in particular, for any polytope game, every CSP $\csp$ in the no-profile-swap-regret menu is a CSP that is implementable against any no-external-regret algorithm. 

\begin{theorem}\label{thm:poly_minimal}
Fix a polytope game $G$. All no-profile-swap-regret algorithms $\cA$ for $G$ share the same asymptotic menu $\cM(\cA) = \cM_{NPSR}$. Moreover, if $\cA$ is a no-regret algorithm, then $\cM_{NPSR} \subseteq \cM(\cA)$.
\end{theorem}
\begin{proof}
We will begin by showing the second part of this theorem: that any no-(external)-regret algorithm $\cA$ for $G$ has the property that $\cM_{NPSR} \subseteq \cM(\cA)$. 

Recall that $\cM_{NPSR}$ is the convex hull of all points of the form $\BR_L(y) \otimes y$. We will show that every point of this form is in $\cM(\cA)$. For every $\hat{x} \otimes \hat{y}$ pair such that $\hat{x} = \BR_L(\hat{y})$, there are two possibilities:
\begin{itemize}
\item $\hat{x}$ is the unique best response to $\hat{y}$. Assume for contradiction that $\hat{x} \otimes \hat{y}$ is not in $\cM(\cA)$. By Theorem~\ref{thm:menu_char}, there must exist at least one point of the form $x \otimes \hat{y}$ in $\cM(\cA)$. Thus, there exists a point of the form $x \otimes \hat{y}$ such that $x$ is not a best response to $\hat{y}$. But this point incurs external regret, and therefore this is a contradiction.   
\item $\hat{x}$ is not the unique best response to $\hat{y}$. Then, note that by our non-degeneracy assumption on $u_{L}$, $\hat{x}$ is not weakly dominated. Thus there must exist a $y^{*} \in \optset$ that uniquely incentivizes $\hat{x}$. Therefore, all points of the form $y(\alpha) = \alpha y^{*} + (1 - \alpha) \hat{y}$ uniquely incentivize $\hat{x}$ for $\alpha \in [0,1]$. Therefore, all points $\hat{x} \otimes y(\alpha)$ are contained within $\cM(\cA)$. As $\cM(\cA)$ is closed, taking the limit of $\alpha \rightarrow 0$ implies that $\hat{x} \otimes y(0) = \hat{x} \otimes \hat{y} \in \cM(\cA)$.  
\end{itemize}

To establish the first part of this theorem, note that by the definition of $\cM_{NPSR}$, if $\cA$ is a no-profile-swap-regret algorithm, then $\cM(\cA) \subseteq \cM_{NPSR}$. On the other hand, since no-profile-swap-regret implies no-external-regret, by the above argument we have that $\cM_{NPSR}  \subseteq \cM(\cA)$. Therefore, $\cM_{NPSR} = \cM(\cA)$.
\end{proof}

\subsubsection{Pareto-optimality}
Finally, we establish the analogue of Theorem~\ref{thm:nfg_pareto}, and show that in polytope games, no-profile-swap-regret algorithms are Pareto optimal. This provides a natural counterpart to Theorem~\ref{thm:nfg_nonmanip} (not only can no-profile-swap-regret algorithms never be manipulated by a strategizing opponent, there are always situations where the learner strictly prefers not to be manipulated).

\begin{theorem}\label{thm:poly_pareto}
Fix a polytope game $G$. Let $\cA$ be a learning algorithm for $G$ with $\CorrSwap(\cA) = \Omega(T)$. Then there exists an optimizer utility $u_O$ such that $V_{L}(\cM_{NPSR}, u_O) > V_{L}(\cM(\cA), u_O)$; i.e., there exists an optimizer against whom it is strictly better to play any no-profile-swap-regret algorithm than $\cA$.
\end{theorem}

The proof of Theorem~\ref{thm:poly_pareto} uses Lemma 4.4 from~\cite{paretooptimal}. While~\cite{paretooptimal} study the setting of normal-form games, in which the strategy polytopes are simplices, their proof of this lemma makes no such assumptions. Given a menu $\cM$, let $\cM^{+}$ represent the set of all extreme points of $\cM$ where $u_{L}$ achieves its maximum value. Furthermore, let $\cM^{-}$ represent the set of all extreme points of $\cM$ where $u_L$ achieves its minimum value.

\begin{lemma}[Lemma 4.4 in~\cite{paretooptimal}]\label{lem:other_paper}
    Let $\cM_{1}$ and $\cM_{2}$ be two distinct asymptotic menus where $\cM_{1}^{+} = \cM_{2}^{+}$. If 
      $\cM_{2} \setminus \cM_{1} \neq \emptyset$, then there exists a $u_{O}$ for which $V_{L}(\cM_{1},u_{O}) > V_{L}(\cM_{2},u_{O})$.
\end{lemma}

We can now prove Theorem~\ref{thm:poly_pareto}.

\begin{proof}[Proof of Theorem~\ref{thm:poly_pareto}]
    Consider an algorithm $\cA$ with menu $\cM = \cM(\cA)$ which incurs $\Omega(T)$ profile swap regret in the worst case. Then, by definition, $\cM \setminus \cM_{NPSR} \neq \emptyset$. Let $\max_{x \in \learnset, y \in \optset}u_{L}(x,y) = u_{L}^{\max}$. By Lemma~\ref{lem:menu_char}, $\cM_{NPSR}$ is the convex hull of all CSPs of the form $x \otimes y$, where $y \in \optset$ and $x \in \BR_{L}(y)$. Thus, $\cM^{+}_{NPSR}$ is the convex hull of all CSPs $x^{*} \otimes y^{*}$ where $u_{L}(x^{*},y^{*}) = u_{L}^{\max}$ (in particular, $\cM^{+}_{NPSR}$ contains every CSP $\csp$ where $u_{L}(\csp) = u_{L}^{\max}$). 
    
    Now, let $\overline{\cM} = \conv(\cM(\cA) \cup \cM^{+}_{NPSR})$. Note that as $\cM(\cA) \subseteq \overline{\cM}$, $\overline{\cM} \setminus \cM_{NPSR} \neq \emptyset$. Furthermore, by construction, $\cM^{+}_{NPSR} = \overline{\cM}^{+}$. Thus, we can invoke Lemma~\ref{lem:other_paper} to show that there exists a $u_{O}$ such that $V_{L}(\cM_{NPSR},u_{O}) > V_{L}(\overline{\cM},u_{O})$. Now, there are two possibilities:
    \begin{itemize}
        \item The extreme point of $u_{O}$ in $\overline{\cM}$ is the same as that in $\cM(\cA)$. Then $V_{L}(\cM(\cA), u_{O}) = V_{L}(\overline{\cM},u_{O}) < V_{L}(\cM_{NPSR}, u_{O})$. 
        \item The extreme point of $u_{O}$ in $\overline{\cM}$ is \emph{not} the same as that in $\cM(\cA)$. Then, the extreme point of $u_{O}$ in $\overline{\cM}$ must belong to $\cM_{NPSR}^{+}$. But for all CSPs $\csp \in \cM_{NPSR}^{+}$, $u_{L}(\csp) = u_{L}^{\max}$, which would imply that $V_{L}(\cM_{NPSR},u_{O}) > u_{L}^{\max}$. This is a contradiction, and therefore this case cannot occur. 
    \end{itemize}
\end{proof}




\section{Minimizing Profile Swap Distance}\label{sec:algorithms}

We now switch our attention to the problem of designing learning algorithms that minimize profile swap regret over finite time horizons. More specifically, we are interested in the following two questions:

\begin{enumerate}
    \item What is the best bound on profile swap regret  that we can guarantee after $T$ rounds? Can we guarantee that this quantity is sublinear in $T$ and polynomial in the dimensions $d_L$ and $d_O$?
    \item Can we construct \emph{computationally efficient} algorithms with these regret bounds? That is, can we construct learning algorithms that run in per-iteration time that is polynomial in $T$, $d_L$, and $d_O$?
\end{enumerate}

Note that for linear swap regret, we have positive answers to both of these questions. In particular, linear swap regret is known to be tractable to minimize: there is a line of work \citep{GordonGreenwaldMarks2008, Farina2023:Polynomial, Farina2024:eah, dann2023pseudonorm} that provides efficient learning algorithms achieving $\poly(d)\sqrt{T}$ linear swap regret in specific classes of games, with \cite{daskalakis2024efficient} now showing that this is possible as long as you have standard oracle access to the learner's action set $\learnset$. On the other side of the spectrum, the best known algorithms for minimizing normal-form swap regret generally incur regret bounds that are either polynomial in the number of \emph{vertices} of $\learnset$ (which easily can be exponential in dimension for Bayesian games and extensive-form games) or scale as $T/(\log T)^{O(1)}$ (i.e., requiring exponential in $1/\eps$ rounds to guarantee $\eps T$ regret). Recent lower bounds \citep{daskalakis2024lowerboundswapregret} show that these rates are necessary for minimizing normal-form swap regret in extensive-form games, implying a mostly negative answer to the above questions.

\subsection{Information-theoretic regret bounds via Blackwell approachability}\label{sec:it-bounds}

In contrast to normal-form swap regret, we will show that it is possible to design learning algorithms that incur at most $O(\sqrt{T})$ profile swap regret. Our main technique is to frame the problem of minimizing profile swap regret as a specific instance of Blackwell approachability where the goal is to approach the no-profile-swap-regret menu $\cM_{NPSR}$.

We therefore begin by briefly reviewing the standard theory of Blackwell approachability. An instance of Blackwell approachability is specified by a convex action set $\learnset$ for the learner, a convex action set $\optset$ for the adversary, a bilinear \emph{vector-valued} payoff function $v: \learnset \times \optset \rightarrow \Rset^{d}$, and a convex target set $\cS \subseteq \Rset^{d}$. The goal of the learner is to run a learning algorithm $\cA$ (outputting a sequence of values $x_t$ in response to the actions $y_1, y_2, \dots, y_{t-1}$ played by the adversary so far) that guarantees that the average payoff $\frac{1}{T}\sum_{t=1}^{T}v(x_t, y_t)$ is close to the target set $\cS$.

In order for this to be possible at all, the function $v$ must be \emph{response-satisfiable} with respect to the target set $\cS$: that is, for every possible action $y \in \optset$ played by the adversary, there must be a response $x(y) \in \learnset$ for the learner with the property that $v(x(y), y) \in \cS$. If $v$ is not response-satisfiable for some $y$, then if the adversary repeatedly plays this action, it is impossible for the learner to force the average payoff to converge to the set $\cS$. Blackwell's approachability theorem states that response-satisfiability is the only necessary requirement -- if $v$ is response-satisfiable with respect to $\cS$, then it is always possible for the learner to steer the average payoff to approach $\cS$. Below we cite a quantitative, algorithmic version of Blackwell's approachability theorem that appears in \cite{mannor2013approachability}.

\begin{theorem}[Blackwell's approachability theorem; Proposition 1 of \cite{mannor2013approachability}]\label{thm:blackwell}
Let $v\colon \cX \times \cY \rightarrow \Rset^{d}$ be a vector-valued bilinear payoff function and $\cS$ be a convex subset of $\Rset^{d}$ such that $v$ is response-satisfiable with respect to $\cS$. Then there exists a learning algorithm $\cA$ (choosing $x_{t} \in \cX$ in response to $y_{1}, y_{2}, \dots, y_{t-1} \in \cY$) with the property that:

$$T \cdot \mathrm{dist}\left(\frac{1}{T}\sum_{i=1}^{t} v(x_t, y_t), \cS\right) \leq 2\norm{v}_{\infty}\sqrt{T},$$
where $\norm{v}_{\infty} = \max_{x \in \cX, y \in \cY}\norm{v(x, y)}_{2}$  and $\mathrm{dist}(z, \cS)$ is the minimum ($\ell_2$) distance from $x$ to the point $\cS$. 

Moreover, this algorithm can be implemented efficiently (in time polynomial in $d$, $\dim(\cX)$, and $\dim(\cY)$) given an efficient separation oracle for the set $\cS$. 
\end{theorem}

There is a clear similarity between the Blackwell approachability problem stated above and the problem of minimizing profile swap regret: whereas in Blackwell approachability, the learner wants to guide the average vector-valued payoff $\frac{1}{T}\sum v(x_t, y_t)$ to converge to the set $\cS$, when minimizing profile swap regret, the learner wants to guide the average CSP $\frac{1}{T}\sum x_t \otimes y_t$ to converge to the no-profile-swap-regret menu $\cM_{NPSR}$. In fact, the function mapping a pair $(x, y)$ of learner and optimizer strategies to the corresponding CSP $x \otimes y$ is itself a vector-valued bilinear payoff function that is response-satisfiable with respect to the no-profile swap regret menu. This allows us to immediately apply Theorem~\ref{thm:blackwell} to get bounds on profile swap regret, which we do below (note that this Theorem will mostly be later subsumed by Theorem~\ref{thm:upper-semi-separation}, which achieves similar bounds but efficiently).

\begin{theorem}\label{thm:upper-approachability}
For any polytope game $G$ (where $\learnset$ and $\optset$ are bounded in unit norm), there exists a learning algorithm $\cA$ that guarantees $\CorrDist(\mathcal{A}^{T}) \leq 2\sqrt{T}$.
\end{theorem}

\begin{proof}
Consider the vector-valued payoff function $v: \learnset \times \optset \rightarrow (\learnset \otimes \optset) \subset \Rset^{d_Ld_O}$ defined via $v(x, y) = x \otimes y$, and let the target set $\cS$ equal the no-profile-swap-regret menu $\cM_{NPSR}$ for the game $G$. Note that $v$ is response-satisfiable for $\cS$, since given any $y \in \optset$, the CSPs in $\BR_{L}(y) \otimes y$ belong to $\cM_{NPSR}$. Note also that $\norm{v}_{\infty} \leq \max_{x \in \cX}\max_{y \in \cY} \norm{x \otimes y} = (\max_{x \in \cX} \norm{x})(\max_{y \in \cY} \norm{y}) \leq 1$, since both $\learnset$ and $\optset$ are bounded in unit norm by assumption.

By Theorem~\ref{thm:blackwell}, there exists a learning algorithm $\cA$ which guarantees that $\mathrm{dist}\left(\frac{1}{T}\sum_{t=1}^{T} v(x_t, y_t), \cS\right) \leq 2\norm{v}_{\infty}\sqrt{T} \leq 2\sqrt{T}$. But $\frac{1}{T}\sum_{t=1}^{T} v(x_t, y_t) = \frac{1}{T}\sum_{t=1}^{T} x_{t} \otimes y_t$ is simply the CSP $\csp$ corresponding to this transcript of play; since $\cS = \cM_{NPSR}$, this distance is exactly the minimum distance from this CSP to the no-profile-swap-regret menu, and is therefore equal to $\CorrDist(\csp)$. It follows that the worst-case profile swap distance $\CorrDist(\mathcal{A}^{T})$ is at most $2\sqrt{T}$. 
\end{proof}

Combining this with our theorem on non-manipulability (Theorem~\ref{thm:poly_nonmanip}), this implies that there exists an $\alpha$-non-manipulable learning algorithm for $\alpha = O(\sqrt{d_Ld_O/T})$.

\subsection{Hardness of computing profile swap regret}

Theorem~\ref{thm:upper-approachability} indicates that it is possible to construct a learning algorithm that achieves good profile swap regret guarantees, but states nothing about the computational complexity of running such an algorithm. The instantiation of Blackwell's approachability in Theorem~\ref{thm:blackwell} is algorithmic, and can be run in polynomial-time as long as one has efficient oracle access to the target set in question. 

In our case, this target set is the no-profile-swap-regret menu $\cM_{NPSR}$. In some cases, we can construct efficient convex oracles for the set $\cM_{NPSR}$; for example, if our polytope game is actually a normal-form game, the menu $\cM_{NPSR}$ reduces to the no-swap-regret polytope $\cM_{NSR}$, which can be written as the explicit intersection of $O(n^2)$ half-spaces when the learner has $n$ actions (i.e., constraints guaranteeing the learner cannot improve their utility by swapping action $i \in [n]$ to action $i' \in [n]$). 

However, we will show that for general polytope games, we cannot hope for any computationally efficient characterization of the no-profile-swap-regret menu (even when we have succinct and efficient descriptions of the sets $\learnset$ and $\optset$). In particular, we will show that even for the setting of Bayesian games, it is impossible to even compute the profile swap distance $\CorrDist(\csp)$ (note that this quantity is simply the distance from $\csp$ to $\cM_{NPSR}$, and would be efficient to compute given standard convex oracles for this set). Our main tool is the following APX-hardness result of computing Stackelberg equilibria in Bayesian games.

\begin{lemma}[Theorem 14 in \cite{MMSSbayesian}]\label{lem:apx-hardness}
It is APX-hard to compute the Stackelberg value for the optimizer in a Bayesian game. That is, there exists a constant $\eps > 0$ such that given a Bayesian game $G$ (with a specific optimizer utility $u_{O}$) and a value $V > 0$ it is NP-hard to distinguish between the cases $\Stack(G, u_O) \leq (1-\eps)V$ and $\Stack(G, u_O) \geq V$. 
\end{lemma}

Lemma~\ref{lem:apx-hardness} almost immediately rules out the possibility of weak optimization oracles for $\cM_{NPSR}$, since by Lemma~\ref{lem:menu_char}, the Stackelberg value $\Stack(G, u_{O})$ is simply the maximum value $u_O$ takes over the menu $\cM_{NPSR}$. Since most weak convex oracles are equivalent (given mild assumptions), this allows us to rule out the existence of any algorithm that can efficiently compute $\CorrDist(\csp)$.

\begin{theorem}\label{thm:hardness}
Define a \emph{profile swap distance oracle} to be an algorithm that takes as input a Bayesian game $G$ (with $\learnset = \Delta_{m}^{c_L}$ and $\optset = \Delta_{n}^{c_O}$), a CSP $\csp \in \learnset \otimes \optset$, and a precision parameter $\delta$, and returns a $\delta$-additive approximation to $\CorrDist(\csp)$. If $P \neq NP$, then there does not exist a profile swap distance oracle that runs in time $\poly(n, m, c_L, c_O, 1/\delta)$.
\end{theorem}
\begin{proof}
Assume to the contrary that such an efficient profile swap distance oracle exists. We first argue that we can use such an algorithm to construct a weak separation oracle for $\cM_{NPSR}$ that also runs in time $\poly(n, m, c_L, c_O, 1/\delta)$.

This follows as a consequence of black-box convex oracle reductions established in \cite{lee2018efficient}. In particular, our profile swap distance oracle is a weak evaluation oracle for the $1$-Lipschitz, convex function $\CorrDist(\csp)$. By Lemma 20 of \cite{lee2018efficient}, it is possible (with polynomially many calls to the original oracle) to extend this to a weak subgradient oracle for the same function. The weak subgradient oracle (when run with precision $\delta$ on input $\csp$) returns a real number $\alpha$ satisfying $|\alpha - \CorrDist(\csp)| \leq \delta$ and a vector $c \in \Rset^{d_Ld_O}$ with the property that $\alpha + c^{T}(\csp' - \csp) \leq \CorrDist(\csp') + \delta$ for all $\csp' \in \learnset \otimes \optset$. 

Now, consider the halfspace given by the inequality $(\alpha - \delta) + c^{T}(\csp' - \csp) \leq 0$. We claim that this is a halfspace that $\delta$-weakly separates $\cM_{NPSR}$ from $\phi$. In particular, every $\csp' \in \cM_{NPSR}$ satisfies $\CorrDist(\csp') = 0$, and thus satisfies this inequality. However, unless $\alpha - \delta \leq 0$, the CSP $\csp' = \csp$ does not satisfy this inequality. But if $\alpha - \delta \leq 0$, then since $|\alpha - \CorrDist(\csp)| \leq \delta$, we must have that $\CorrDist(\csp) \leq 2\delta$. We have therefore either provided a hyperplane separating $\cM_{NPSR}$ from $\phi$, or certified that $\phi$ is distance at most $2\delta$ from $\cM_{NPSR}$, completing the construction of our weak separation oracle.

With a weak separation oracle, we can construct a weak optimization oracle, and approximately optimize any linear function over $\cM_{NPSR}$. But by Lemma~\ref{lem:menu_char}, the optimal value of the linear function $u_{O}$ over $\cM_{NPSR}$ is $\Stack(G, u_O)$. Thus, this oracle would allow us to distinguish between the two cases in Lemma~\ref{lem:apx-hardness} in polynomial time, contradicting the APX-hardness of the result unless $P = NP$.
\end{proof}

\subsection{Efficient learning algorithms via semi-separation}\label{sec:alg-semisep}

Nevertheless, despite the hardness of computing profile swap regret (and thus of running the algorithm of Theorem \ref{thm:upper-approachability}), we will now show that it is still possible to design efficient learning algorithms $\cA$ that obtain profile swap regret guarantees of the form $\poly(d_L, d_O)\sqrt{T}$. 

To describe the idea behind semi-separation, it will be convenient to express the problem of Blackwell approachability in slightly different language. Previously, we described an approachability instance as being parameterized by (in addition to the two action sets of the learner and adversary) a vector-valued bilinear payoff function $v$ and a target set $\cS$, with the idea that the learner is trying to steer the average vector-valued payoff to be close to $\cS$. In learning-theoretic contexts, it is frequently mathematically more convenient to work with another formulation of approachability that we will refer to as \emph{orthant-approachability}. An \emph{orthant-approachability} instance is parameterized by a single \emph{convex set} $\cU$ of \emph{one-dimensional} biaffine payoff functions $u: \learnset \times \optset \rightarrow \Rset$. The goal of the learner is now to minimize the worst-case average value of $u(x_t, y_t)$ over \emph{all} biaffine functions $u \in \cU$; that is, they wish to minimize the approachability loss given by $\AppLoss(\bx, \by) = \max_{u \in \cU} \frac{1}{T}\sum_{t=1}^{T} u(x_t, y_t)$.

In particular, the learner would like to guarantee that $\AppLoss(\bx, \by)$ grows sublinearly in $T$. In order for this to be possible, $\cU$ must contain only \emph{response-satisfiable} functions $u$. In the language of orthant-approachability, this means that for any $y \in \cY$, there exists an $x \in \cX$ such that $u(x, y) \leq 0$ for all $u \in \cU$ (we call this ``orthant''-approachability since this can be thought of approaching the negative orthant for the infinite-dimensional bilinear payoff function whose entries are given by $u(x, y)$ for different $u \in \cU$).

It can be shown that orthant-approachability and our previous formulation of approachability are essentially equivalent (in fact, orthant-approachability is slightly more general, in that it can capture other distance metrics and some infinite-dimensional bilinear payoffs). In particular, we show below (following an argument in \cite{dann2024rate}) that any instance of our previous formulation of approachability can be equivalently expressed as an orthant-approachability problem. 

\begin{lemma}\label{lem:orthant-approach}
Let $v: \learnset \times \optset \rightarrow \Rset^{d}$ be a vector-valued bilinear payoff function that is response-satisfiable with respect to the target set $\cS \subseteq \Rset^{d}$. There exists a convex subset $\cU$ of response-satisfiable biaffine functions $u: \learnset \times \optset \rightarrow \Rset$ such that, for any sequences $x_t \in \learnset$ and $y_t \in \optset$,

$$\max_{u \in \cU} \frac{1}{T}\sum_{t=1}^{T} u(x_t, y_t) = \mathrm{dist}\left(\frac{1}{T}\sum_{i=1}^{t} v(x_t, y_t), \cS\right)$$
\end{lemma}
\begin{proof}
For each unit vector $h \in \Rset^{d}$, consider the biaffine function $u_{h}(x, y) = \langle v(x, y), h\rangle - \sigma_{h}(\cS)$, where $\sigma_{h}(\cS) = \max_{z \in \cS}\langle h, z\rangle$. Let $\cU$ be the convex hull of all such functions $u_{h}$. 

If $\frac{1}{T}\sum_{t=1}^{T} u_{h}(x_t, y_t) \geq R$, then $\left\langle \frac{1}{T}\sum_{t=1}^{T} v(x_t, y_t), h \right\rangle - \sigma_{h}(\cS) \geq R$, and therefore $\mathrm{dist}\left(\frac{1}{T}\sum_{i=1}^{t} v(x_t, y_t), \cS\right)$ is at least $R$. On the other hand, if $\mathrm{dist}\left(\frac{1}{T}\sum_{i=1}^{t} v(x_t, y_t), \cS\right)$ is at least $R$, then there must be a direction $h$ in which $\left\langle \frac{1}{T}\sum_{t=1}^{T} v(x_t, y_t), h \right\rangle - \sigma_{h}(\cS) \geq R$, and therefore we have $\frac{1}{T}\sum_{t=1}^{T} u_{h}(x_t, y_t) \geq R$.
\end{proof}

As with classical Blackwell approachability, there exist efficient learning algorithms for minimizing approachability loss in orthant-approachability -- however, just as the earlier algorithms required separation oracles for the target set $\cS$, the algorithms for orthant-approachability require separation oracles for the set of payoff functions $\cU$. And likewise, just as Theorem~\ref{thm:hardness} precludes the existence of efficient separation oracles for the no-profile-swap-regret menu $\cM_{NPSR}$, it also precludes the existence of separation oracles for the corresponding convex set $\cU$.

\cite{daskalakis2024efficient} bypass a similar obstacle for minimizing linear swap regret via a technique they call \emph{semi-separation}. We introduce a version of this same technique here for orthant-approachability. To motivate this technique, note that if we want to achieve orthant-approachability with respect to a specific (perhaps intractible) set of biaffine functions $\cU$, it suffices to achieve orthant-approachability with respect to any superset $\cU'$ containing $\cU$ (since the maximum over all $u'$ in $\cU'$ will always be at least as large as the maximum over all $u$ in $\cU$). In particular, if we could find a superset $\cU'$ of $\cU$ with more tractable convex oracles, then we could simply run orthant-approachability over that superset.

Unfortunately, it is not too hard to show that the minimality of the menu $\cM_{NPSR}$ implies that no superset $\cU'$ of $\cU$ can contain only response-satisfiable functions. The idea of \cite{daskalakis2024efficient} is to run an approachability algorithm on some tractable superset $\cU'$ anyway, despite $\cU'$ not being an approachable set. The approachability algorithm will either work without issue (in which case it will provide our desired guarantee on the approachability loss), or at some point the algorithm will identify a $u' \in \cU'$ that is not response-satisfiable, and thus not in $\cU$. If we can then construct a separating hyperplane that separates $u'$ from $\cU$, we can use this to refine  our superset $\cU'$. This algorithm -- which takes a non-response-satisfiable $u$, and returns a hyperplane separating $u$ from $\cU$ -- is what we refer to as a semi-separation oracle (unlike a standard separation oracle, it does not work for all $u \not \in \cU$, but only non-response-satisfiable $u$).



Formally, a \emph{semi-separation oracle} for a set of response-satisfiable payoff functions $\cU$ is an algorithm which takes an arbitrary bi-affine function $u: \learnset \times \optset \rightarrow \Rset$ and either (i) returns that the function $u$ is response-satisfiable, or (ii) returns a hyperplane separating $u$ from $\cU$. The following theorem states that efficient semi-separation oracles suffice for performing orthant-approachability efficiently. The proof closely follows the proof of \cite{daskalakis2024efficient} for the case of linear swap functions, and is deferred to Appendix \ref{app:semi-sep}.


\begin{theorem}[Approachability with semi-separation oracles]\label{thm:approach-semisep}
Consider an orthant-approachability instance $(\cX, \cY, \cU)$ where $\cX \subseteq \cB(D_{\cX})$, $\cY \subseteq \cB(D_{\cY})$, and $\cB(u_{0}, \rho) \subseteq \cU \subseteq \cB(D)$ (for some known $u_{0}$ and radius $\rho$). If we have access to efficient separation oracles for $\cX$ and $\cY$, and an efficient semi-separation oracle for the set $\cU$, then Algorithm~\ref{algo:semisep-approach} has the guarantee that

$$\max_{u \in \cU} \frac{1}{T}\sum_{t=1}^{T} u\left(x_t, y_t \right) \leq O\left(\frac{D_{\cX}D_{\cY}D}{\sqrt{T}}\right),$$

\noindent
and runs in time $\poly(D_{\cX}, D_{\cY}, D, \rho^{-1}, \dim(\cU), T)$.
\end{theorem}

In the remainder of this section, we will show that we can construct a semi-separation oracle for the approachability problem corresponding to minimizing profile swap regret, thus allowing us to efficiently minimize it via Theorem~\ref{thm:approach-semisep}. In particular, consider the set $\cU$ containing all bi-affine functions $u: \learnset \times \optset \rightarrow \Rset$ of the form $u(x, y) = \langle h, x \otimes y \rangle - b$ that satisfy: i. $\norm{h} \leq 1$, ii. $|b| \leq 1$, and iii. $u(x, y) \leq 0$ for all $x \in \BR_{L}(y)$. Note that this is a convex set of bi-affine functions, is response-satisfiable (via the third constraint), and for any CSP $\csp$ satisfies $\max_{u \in \cU} u(\csp) = \CorrDist(\csp)$ (in particular, it is equivalent to the set constructed in the proof of Lemma~\ref{lem:orthant-approach}). We produce an efficient semi-separation oracle for this choice of $\cU$.

\begin{lemma}\label{lem:semi-sep-exist}
Given efficient separation oracles for the sets $\learnset$ and $\optset$, it is possible to construct an efficient semi-separation oracle for the set $\cU$ defined above.
\end{lemma}
\begin{proof}
Consider an arbitrary bi-affine function $u_{\text{in}}: \learnset \times \optset \rightarrow \Rset$ that our semi-separation oracle receives as input. Compute (e.g., via linear programming) the minimax value $V = \min_{x \in \learnset} \max_{y \in \optset} u_{\text{in}}(x, y)$.

If $V \leq 0$, then this function is response-satisfiable (i.e., there exists an $x^*$ such that $u_{\text{in}}(x^*, y) \leq 0$ for all $y \in \optset$), and we can return that $u$ is response-satisfiable. Otherwise, there exists a $y^* \in \optset$ such that $u_{\text{in}}(x,  y^*) \geq V > 0$ for all $x \in \learnset$, which we can again compute efficiently via a linear program. 

Now, pick an arbitrary $x_{y^*} \in \BR_{L}(y^*)$. Note that, by construction, $u(x_{y^*}, y^*) \leq 0$ for all $u \in \cU$; on the other hand, for our queried $u_{\text{in}}$, we have that $u_{\text{in}}(x_{y^*}, y^*) \geq V > 0$. It follows that the linear constraint $u(x_{y^*}, y^*) \leq 0$ is a linear constraint on bilinear functions $u$ separating $u_{\text{in}}$ from the set $\cU$. 
\end{proof}

From this semi-separation oracle and Theorem~\ref{thm:approach-semisep}, we immediately obtain an efficient algorithm for minimizing profile swap regret. 

\begin{theorem}\label{thm:upper-semi-separation}
Given efficient separation oracles for $\cX$ and $\cY$, there exists a learning algorithm $\cA$ such that $\CorrDist(\mathcal{A}^{T}) = O(\sqrt{T})$ that runs in $\poly(T, d_L, d_O)$ time per iteration.
\end{theorem}
\begin{proof}
Given Lemma~\ref{lem:semi-sep-exist}, it suffices to check the conditions of Theorem~\ref{thm:approach-semisep} for the set $\cU$ defined above. By assumption, our $\learnset$ and $\optset$ are bounded within the unit ball (so we can take $D_{\learnset} = D_{\optset} = 1$). The maximum norm of any element in our set $\cU$ is at most $\sqrt{||h||^{2} + |b|^2} \leq \sqrt{2}$, so we can take $D = \sqrt{2}$. Finally, note that any $u$ with $||h|| \leq 1/2$ and $b \geq 1/2$ will satisfy $u(x, y) \leq 0$ for any $x \in \learnset$ and $y \in \optset$; we can therefore take $u_{0}(x, y) = -3/4$ (i.e., the $u$ specified by $h = 0$ and $b = 3/4$) and $\rho = 1/4$.
\end{proof}


\section{Game-Agnostic Learning and Polytope Swap Regret}\label{sec:game-agnostic}

In our discussion of online learning thus far, we have assumed that our learning algorithms $\cA$ are tailored to a specific polytope game $G$. By this, we mean that our algorithms operate under the assumption that the learner can see the sequence of actions $y_t \in \optset$ played by the optimizer so far and that the learner can compute the game-specific payoff function $u_L(x, y)$. These assumptions are typical for applications of online learning to games (for example, when such algorithms are used for equilibrium computation in extensive-form games).

However, many algorithms in the field of adversarial online learning are designed so that they require only the counterfactual rewards (alternatively, losses) faced by the learner over time. That is, instead of being told the mixed strategy $y_t$ the optimizer played in round $t$, it is enough for these learning algorithms to have access to the linear function $r_t(x)$ describing what utility the learner would have received if they had played action $x \in \learnset$ that round. Such algorithms can be used by a learner to play any repeated game where the learner has action set $\learnset$, regardless of their specific payoff function $u_L$ or the optimizer's action space $\optset$. In this section, we explore what happens when we restrict ourselves to  game-agnostic learning algorithms -- we will show that profile swap regret is an inherently game-dependent notion, and if we wish to minimize profile swap regret with a game-agnostic learning algorithm, we must actually minimize polytope swap regret.

Formally, a \emph{game-agnostic learning algorithm} $\cA$ is defined in the same way as our earlier (game-dependent) learning algorithms, with the only difference being that in each horizon-dependent learning algorithm, $\cA_{t}^{T}$ now maps a sequence of affine linear \emph{reward functions} $r_1, r_2, \dots, r_{t-1}$ to the action $x_{t} \in \learnset$ the learner will take at time $t$. Each $r_t$ is an affine linear function  sending $\learnset$ to $[-1, 1]$; we will write $\learnset^{*}$ to denote the set of such functions (so, each $r_{t}$ belongs to $\learnset^*$). 

When used to play a specific polytope game $G$, the function $r_t$ is constructed via $r_{t}(x) = u_L(x, y_t)$, where $y_t \in \optset$ is the action of the optimizer at round $t$. In particular, any transcript $(\bx, \by)$ of a polytope game $G$ corresponds to a \emph{game-agnostic transcript} $(\bx, \br)$ via this mapping. The game-agnostic transcript is sufficient to compute many quantities relevant to the learner, including their total utility over the course of the game and several variants of regret, as shown in the following lemma.

\begin{lemma}\label{lem:game-agnostic-regret}
\sloppy{Let $(\bx, \by)$ be the transcript of some repeated polytope game $G$. It is possible to compute external regret $\Reg(\bx, \by)$, linear swap regret $\LinSwap(\bx, \by)$, and polytope swap regret $\PolySwap(\bx, \by)$ from the game-agnostic transcript $(\bx, \br)$ corresponding to $(\bx, \by)$ (without any knowledge of $G$ or $u_L$ aside from the learner's action set $\learnset$).}
\end{lemma}
\begin{proof}
Note that the only way in which the optimizer's actions $y_t$ are used in the definitions of external regret $\Reg(\bx, \by)$, linear swap regret $\LinSwap(\bx, \by)$, and polytope swap regret $\PolySwap(\bx, \by)$ is to compute quantities of the form $u_{L}(\cdot, y_t)$. Since $r_t(x) = u_L(x, y_t)$, we can compute all such quantities given access to $r_t(x)$.
\end{proof}

As a consequence of Lemma~\ref{lem:game-agnostic-regret}, we can define quantities such as $\Reg(\bx, \br)$, $\LinSwap(\bx, \br)$, and $\PolySwap(\bx, \br)$ (defining these swap regrets as functions of the game-agnostic transcript). Similarly, we can define e.g. $\PolySwap(\cA^{T})$ to be the maximum polytope swap regret incurred by the game-agnostic (horizon $T$) algorithm $\cA^{T}$ against the worst-case sequence of reward functions $\br$.

\sloppy{In contrast to the result of Lemma~\ref{lem:game-agnostic-regret}, it is \emph{not possible} to compute the profile swap regret $\CorrSwap(\bx, \by)$ in this way: there exist pairs of transcripts $(\bx_1, \by_1), (\bx_2, \by_2)$ for polytope games with the same action set $\learnset$ that have the same game-agnostic transcript $(\bx, \br)$, but where $\CorrSwap(\bx_1, \by_1) \neq \CorrSwap(\bx_2, \by_2)$. That is, unlike the other variants of swap regret we have discussed, profile swap regret is fundamentally game-dependent (intuitively, this is because the space of possible decompositions of the CSP $\csp$ depends on the optimizer's action space $\optset$). We give an explicit example of this phenomenon in Appendix \ref{app:prof-from-transcript}.}

Instead, for any specific polytope game $G$ and game-agnostic learning algorithm $\cA$ (where $\cA$ and $G$ share the same learner action set $\learnset$), we can define $\CorrSwap_{G}(\cA^{T})$ to be the maximum profile swap regret a learner incurs by using $\cA^{T}$ to select their actions in the repeated game $G$. We will show that if we want to bound the profile swap regret $\CorrSwap_{G}(\cA^{T})$ for all polytope games $G$, we must bound the (game-agnostic) polytope swap regret $\PolySwap(\cA^{T})$. Our main tool is the following lemma, showing that we can always instantiate any game-agnostic transcript as a transcript of an actual polytope game in a way that makes profile swap regret and polytope swap regret agree.

\begin{lemma}\label{lem:agnostic-corr-to-poly}
\sloppy{Let $(\bx, \br)$ be a game-agnostic transcript. There exists a transcript $(\bx, \by)$ of a polytope game $G$ where $(\bx, \br)$ is the game-agnostic transcript corresponding to $(\bx, \by)$ and $\CorrSwap(\bx, \by) = \PolySwap(\bx, \by)$.}
\end{lemma}
\begin{proof}
Let $\bx = (x_1, x_2, x_3, \dots, x_T) \in \learnset^{T}$ and $\br = (r_1, r_2, \dots, r_T) \in (\learnset^{*})^T$. Consider the polytope game $G$ where the action set $\optset$ of the optimizer is $\Delta_{T}$ (distributions over $T$ pure actions), and the learner's utility function is given by $u_{L}(x, y) = \sum_{i=1}^{T} y_i r_i(x)$. Note that if we let $\by = (e_1, e_2, \dots, e_T)$ (where $e_i$ is the $i$th unit vector), then $(\bx, \br)$ is the game-agnostic transcript corresponding to $(\bx, \by)$.

We now argue that $\CorrSwap(\bx, \by) = \PolySwap(\bx, \by)$. The average CSP $\csp$ of the transcript $(\bx, \by)$ is given by

\begin{equation}\label{eq:agnostic-decomp1}
\csp = \frac{1}{T}\left(\sum_{t=1}^{T} x_t \otimes e_t \right).
\end{equation}

\noindent
Now, consider any decomposition of $\csp$ of the form

\begin{equation}\label{eq:agnostic-decomp2}
\csp = \sum_{v \in \learnvert} \lambda_{v} (v \otimes y_{v}),
\end{equation}

\noindent
where $\lambda_{v} \geq 0$, $\sum_{v} \lambda_v = 1$, and $y_{v} \in \optset$. We claim that, for any such  decomposition and any $t \in [T]$, we must have that $\sum_{v} \lambda_{v}y_{v, t} = 1/T$. To see this, consider the outcome when the (bi-affine) linear function $\rho: \learnset \otimes \optset \rightarrow \Rset$ defined via $\rho(x \otimes y) = y_{t}$ is evaluated on the (equal) right hand sides of both \eqref{eq:agnostic-decomp1} and \eqref{eq:agnostic-decomp2}; for \eqref{eq:agnostic-decomp1}, this equals $1/T$, and for \eqref{eq:agnostic-decomp2}, this equals $\sum_{v} \lambda_{v}y_{v, t}$. 

Given this, for each $t$, define $x_{t}^{V} \in \Delta(\learnvert)$ via $x_{t, v}^{V} = T\lambda_{v} y_{v, t}$. Note that since $\sum_{v} x_{t, v}^{V} = T\sum_{v} \lambda_{v}y_{v, t} = 1$, $x_{t}^{V}$ does indeed belong to $\Delta(\learnvert)$. But now, from the definition of polytope swap regret, this specific action decomposition implies that

\begin{eqnarray*}
\PolySwap(\bx, \by) &\leq & \max_{\pi: \learnvert \rightarrow \learnvert} \left(\sum_{t=1}^{T} u_L(\overline{\pi}(\xV_t), y_t) - u_L(x_t, y_t) \right) \\
&=& T \cdot \max_{\pi: \learnvert \rightarrow \learnvert} \left(\sum_{v \in \learnvert} \lambda_{v} u_L(\overline{\pi}(v), y_v) - u_L(v, y_v) \right).
\end{eqnarray*}

Since $\CorrSwap(\bx, \by)$ is equal to the minimum of this final expression over all possible decompositions of the form \eqref{eq:agnostic-decomp2}, this implies that $\PolySwap(\bx, \by) \leq \CorrSwap(\bx, \by)$. Together with the fact that $\CorrSwap(\bx, \by) \leq \PolySwap(\bx, \by)$ (Theorem~\ref{thm:comp-prof-swap}), we have that $\PolySwap(\bx, \by) = \CorrSwap(\bx, \by)$.
\end{proof}

With Lemma~\ref{lem:agnostic-corr-to-poly}, we can prove the main result of this section.

\begin{theorem}\label{thm:agnostic-main}
Let $\cA$ be a game-agnostic learning algorithm. Then for any $T > 0$, $\PolySwap(\cA^{T}) = \max_{G} \CorrSwap_{G}(\cA^{T})$, where this maximum is over all polytope games $G$ with the same learner action set $\learnset$ as $\cA$. 
\end{theorem}

\begin{proof}
For any polytope game $G$ and transcript $(\bx, \by)$ in $G$, Theorem~\ref{thm:comp-prof-swap} tells us that $\CorrSwap(\bx, \by) \leq \PolySwap(\bx, \by)$. It follows that $\CorrSwap_{G}(\cA^{T}) \leq \PolySwap(\cA^{T})$ for any polytope game $G$ (with the same learner action set $\learnset$ as $\cA$). 

Conversely, consider any game-agnostic transcript $(\bx, \br)$ of length $T$ produced by $\cA^{T}$. By Lemma~\ref{lem:agnostic-corr-to-poly}, there exists a polytope game $G$ and a (game-dependent) transcript $(\bx, \by)$ corresponding to $(\bx, \br)$ with the property that $\CorrSwap(\bx, \by) = \PolySwap(\bx, \by)$. But now, note that if the optimizer plays the sequence of actions $\by$ against a learner employing the game-agnostic algorithm $\cA$, the learner will see the sequence of reward functions $\br$ and therefore play the sequence of actions $\bx$. It follows that there exists a polytope game $G$ where $\CorrSwap_{G}(\cA^{T}) \geq \PolySwap(\cA^{T})$. Together with the above result, this implies the theorem statement.
\end{proof}

\begin{remark}
It is interesting to discuss the results of \cite{rubinstein2024strategizing} in light of the above discussion. \cite{rubinstein2024strategizing} answer an open question of \cite{MMSSbayesian} by showing that minimizing \emph{polytope swap regret} is both necessary and sufficient for implying non-manipulability in Bayesian games. However, the original question of \cite{MMSSbayesian} was posed for (a slight variant of) game-agnostic learning algorithms -- indeed, \cite{rubinstein2024strategizing} point out that without the ability to change the Bayesian game (by increasing the number of actions for the optimizer), the conjecture is not true (we use a similar construction to separate profile swap regret and polytope swap regret in Theorem~\ref{thm:comp-prof-swap}). Our Theorem~\ref{thm:agnostic-main} above (along with the characterization of non-manipulability in Theorem~\ref{thm:poly_nonmanip}) can therefore be thought of as a generalization of this result of \cite{rubinstein2024strategizing} to all polytope games. 
\end{remark}

It is tempting to try to define a game-agnostic analogue of profile swap regret by considering the game $G^{*}$ where the optimizer's actions set $\optset$ is the set $\learnset^{*}$ of reward functions, and where $u_L(x, r) = r(x)$ for any $x \in \learnset$, $r \in \learnset^{*}$. We will write $\CorrSwap_{G^{*}}(\bx, \br)$ to be the profile swap regret of the learner in this game, and call this quantity the \emph{game-agnostic profile swap regret} of this transcript. 

Since $\dim(\learnset^{*}) = \dim(\learnset)$, our later constructions for efficient no-profile-swap-regret learning algorithms will also produce efficient game-agnostic learning algorithms that minimize game-agnostic profile swap regret. It would be nice if game-agnostic profile swap regret was an upper bound on (game-specific) profile swap regret, thus giving us a way to efficiently bound polytope swap regret. Unfortunately, this is not the case.

\begin{theorem}\label{thm:game-agnostic-profile-bad}
There exists a polytope game $G$ and a transcript $(\bx, \by)$ in $G$ such that

$$\CorrSwap_{G}(\bx, \by) > \CorrSwap_{G^*}(\bx, \br).$$
\end{theorem}
\begin{proof}
See Appendix \ref{app:game-agnostic-profile-bad} for a proof. We build off the example used to demonstrate profile swap regret is not game-agnostic in Appendix \ref{app:prof-from-transcript}.
\end{proof}

However, there is a weaker sense in which this is true. We say that a polytope game $G$ is \emph{optimizer full-rank} if the map that sends an optimizer action $y \in \optset$ to its corresponding reward function $r \in \learnset^{*}$ is injective -- i.e., we can fully recover the optimizer action from the reward. Note in particular that this requires the optimizer's action set to have dimension at most as large as that of the learner's action set. Under this constraint, game-agnostic profile swap regret does indeed upper bound the profile swap regret of the game.

\begin{theorem}\label{thm:opt-full-rank}
Let $G$ be a polytope game that is optimizer full-rank. Then if $(\bx, \br)$ is the game-agnostic transcript corresponding to the transcript $(\bx, \by)$ of $G$, we have that

$$\CorrSwap_{G}(\bx, \by) \leq \CorrSwap_{G^*}(\bx, \br).$$
\end{theorem}
\begin{proof}
Let $\csp^{*}$ be the ``game-agnostic'' CSP $\csp^{*} = (1/T) \sum_{t} (x_t \otimes r_t) \in \learnset \otimes \learnset^{*}$ of the transcript $(\bx, \br)$ in the game $G^{*}$. Consider the decomposition of $\csp^*$ in the manner of \eqref{eq:csp-decomp}, i.e.,

$$\csp^{*} = \sum_{k=1}^{K} \lambda_{k} (x_{(k)} \otimes r_{(k)}),$$

\noindent
which realizes $\CorrSwap_{G^*}(\bx, \br)$ (so, $\CorrSwap_{G^*}(\bx, \br) = \sum_{k} \lambda_{k} \Reg(x_{(k)}, r_{(k)})$). 

Because $G$ is optimizer full-rank, there exists a linear function $M: \learnset^{*} \rightarrow \optset$ such that $M(r)$ is the unique element in $\optset$ with reward function $r$ (so $M(r_t) = y_t$ for all $t \in [T]$). For each $k \in [K]$, let $y_{(k)} = M(r_{(k)})$. Note that by applying the linear map $x \otimes r \rightarrow x \otimes M(r)$ to the equality

$$\frac{1}{T} \sum_{t} (x_t \otimes r_t) = \sum_{k=1}^{K} \lambda_{k} (x_{(k)} \otimes r_{(k)}),$$

\noindent
we find that the CSP $\csp = (1/T) \sum_{t} (x_t \otimes y_t) = (1/T) \sum_{t} (x_t \otimes M(r_t))$ of the original transcript $(\bx, \by)$ satisfies

$$\csp = \sum_{k=1}^{K} \lambda_{k} (x_{(k)} \otimes y_{(k)}).$$

But from this decomposition, we have that $\CorrSwap_{G}(\bx, \by) \leq \sum_{k=1}^{K}\lambda_{k} \Reg_{G}(x_{(k)}, y_{(k)}) = \sum_{k=1}^{K}\lambda_{k} \Reg_{G^*}(x_{(k)}, r_{(k)}) = \CorrSwap_{G^*}(\bx, \br)$. 
\end{proof}

Can we extend Theorem~\ref{thm:opt-full-rank} to games beyond the class of optimizer-full-rank games (possibly at the cost of increasing the dimension of the optimizer's action set $\optset$)? For example, can we game-agnostically bound profile swap regret in all games where the optimizer has at most $n$ actions, or that are ``almost'' optimizer-full-rank (e.g., there is a low-dimensional subspace of possible actions implementing a given reward)? We leave this as an interesting future direction.

\section{Implications for Equilibrium Computation}\label{sec:equilibria}

We finally return to the question of equilibrium computation. Indeed, the original motivation for introducing swap regret was to design learning dynamics that converge to the notion of correlated equilibria in normal-form games (e.g., see \cite{foster1997calibrated}). It seems natural then, that when defining swap regret for polytope games, we should choose a quantity whose minimization guarantees convergence to the set of correlated equilibria.

The problem here is that, just as it is not exactly clear what the definition of swap regret should be in polytope games, it is also not exactly clear what the definition of correlated equilibrium should be in polytope games. Even for the restricted set of Bayesian games, existing correlated equilibrium notions include linear correlated equilibria, normal-form correlated equilibria, agent-normal-form correlated equilibria, and communication equilibria \citep{fujii2023bayes}. This raises the question of why we might want to compute correlated equilibria in the first place -- what properties might we desire from its definition?

One motivation for studying correlated equilibria in normal-form games, arising from the original definition in \cite{aumann1974subjectivity}, is that we can view a correlated equilibrium as an outcome that can be implemented by a third-party \emph{mediator} who privately recommends an action to each player. If no one can gain by deviating from these suggestions, the resulting strategy profile constitutes a correlated equilibrium. This gives a definition of correlated equilibrium that is particularly amenable to mechanism design (the classic example here is that of a traffic light coordinating the actions of many cars at an intersection). 

This mediator-based definition is relatively straightforward to extend to the setting of a polytope game $G$ -- we are looking for distributions over joint recommendations $(x, y) \in \learnset \otimes \optset$ with the property that neither player has an incentive to deviate after seeing their recommendation. It can be shown (see Appendix \ref{app:mediator-nf}) that these recommendations can, without loss of generality, be supported on the vertices of $\learnset$ and $\optset$, and so these mediator-based correlated equilibria are exactly the \emph{normal-form correlated equilibria (NFCE)} of $G$ (equivalently, these are the correlated equilibria of the normal-form vertex game $G^{V}$). Formally, an NFCE of $G$ is a vertex game CSP $\cspV \in \Delta(\learnvert \times \optvert)$ with the property that\footnote{Since in this section we are concerned with settings where both players are learning, instead of referring to the two players as the ``learner'' and ``optimizer'', we will refer to them as the ``$\cX$-player'' and ``$\cY$-player'', adjusting subscripts accordingly.} $\NormSwap_{X}(\cspV) = \NormSwap_{Y}(\cspV) = 0$.  Since this game can be thought of as a normal-form game over the action set simplices $\Delta(\learnvert)$ and $\Delta(\optvert)$,  $\NormSwap_{X}(\cspV)$ is simply defined as the swap-regret of a CSP $\cspV$ in this game.

But a second, learning-theoretic, motivation for studying correlated equilibria is that they directly represent possible summaries of the outcomes of learning dynamics in repeated games. These are useful for understanding what possible outcomes we might expect from repeated multi-agent interactions in strategic settings -- after all, it is now generally accepted that it is intractable (computationally and otherwise) for players to play according to the Nash equilibrium of an arbitrary general-sum game, and instead it is more reasonable to model players as performing some form of learning to decide their actions over time.

From this perspective, the definition of correlated equilibrium should follow from the definition of swap regret (specifically, whichever definition of swap regret we choose to model rational behavior in repeated games). For profile swap regret, this gives rise to the notion of \emph{profile correlated equilibria (PCE)}. In particular, a CSP $\csp$ is a PCE of a polytope game $G$ if $\CorrSwap_{X}(\csp) = \CorrSwap_{Y}(\csp) = 0$.

In normal-form games, these two motivations give rise to exactly the same notion of correlated equilibrium (and hence this distinction often goes unmade). In polytope games these two notions are not immediately comparable -- a normal-form CE $\cspV$ and a profile CE $\csp$ belong to different sets and have different dimensions. However, every normal-form CE $\cspV \in \Delta(\learnvert \times \optvert)$ naturally corresponds to a CSP $\Proj(\cspV)$ obtained by sending $\cspV$ to the element $\sum_{v \in \learnvert, w \in \optvert} \cspV(v, w) (v \otimes w)$. This raises the question: are profile correlated equilibria exactly the CSPs corresponding to valid normal-form correlated equilibria?

To see why we might expect this to be the case for profile swap regret in particular, we point out that a one-sided version of this question has a positive answer. In particular, every CSP $\csp$ that incurs zero profile swap regret for the $\cX$-player (i.e., satisfies $\CorrSwap_{\cX}(\csp) = 0$) can be instantiated as $\Proj(\cspV)$ for some vertex-game CSP $\cspV$ that incurs zero normal-form swap regret for the $\cX$-player (see Appendix \ref{app:one-sided}). From a mediator point-of-view, this means that any CSP $\csp$ the $\cX$-player encounters by playing a no-profile-swap-regret learning algorithm can also be induced by a ``one-sided mediator'' that only needs to consider the incentives of the $\cX$-player (i.e., one that the $\cY$-player will blindly trust). We note that this property follows from the fact that the no-profile-swap-regret menu is minimal (Theorem~\ref{thm:poly_minimal}), and does not hold for weaker forms of regret like linear swap regret.

However, we prove that the answer to the original question is no. Specifically, we show that although every CSP $\Proj(\cspV)$ corresponding to a normal-form CE $\cspV$ is in fact a profile CE, there exist profile CE that cannot be written in this form. In fact, we go further and show that there are profile CE with utility profiles that no normal-form CE can generate.

\begin{theorem}\label{thm:equilibrium-gap}
    If $\cspV$ is a normal-form CE in a polytope game $G$, then $\Proj(\cspV)$ is a profile CE in $G$. On the other hand, there exists a polytope game $G$ and a profile CE $\csp$ such that there does not exist a normal-form CE $\cspV$ satisfying $\Proj(\cspV) = \csp$ (in fact, there does not even exist a normal-form CE $\cspV$ where $u_X(\cspV) = u_X(\csp)$ and $u_Y(\cspV) = u_Y(\csp)$).
\end{theorem}

In particular, Theorem~\ref{thm:equilibrium-gap} means that there are outcomes of repeated polytope games (i.e., CSPs $\csp$) for which both players can imagine as implementable by a one-sided mediator, but for which no (two-sided) mediator protocol can actually induce. This is a fundamental property of polytope games that extends beyond the specific definition of profile swap regret, and distinguishes general polytope games from normal-form games.

Of course, we can also ask for which polytope games $G$ does a separation exist like that in Theorem~\ref{thm:equilibrium-gap} -- after all, there is no separation for the case of normal-form games, which are special cases of polytope games. Fully characterizing this is an interesting open question, but below we show that this gap disappears whenever \emph{either} player's action set is a simplex. Notably, this covers all standard Bayesian games where only one of the two players has any private information (e.g., some auction/pricing games). 

\begin{theorem}\label{thm:equiv-simplex}
Let $G$ be a polytope game where $\cX$ is a simplex (but where $\cY$ may be any polytope). Then, if $\csp$ is a profile CE in $G$, there exists a normal-form CE $\cspV$ for which $\Proj(\cspV) = \csp$.
\end{theorem}
\begin{proof}
Since $\CorrSwap_{X}(\csp) = 0$, we can write

\begin{equation}\label{eq:equiv-cspx}
\csp = \sum_{v \in \learnvert} \lambda_{v} (v \otimes y_{v}).
\end{equation}

\noindent
for some choice of $y_{v} \in \optset$ and nonnegative $\lambda_{v}$ summing to one. Similarly, since $\CorrSwap_{Y}(\csp) = 0$, we can write

\begin{equation}\label{eq:equiv-cspy}
\csp = \sum_{w \in \optvert} \mu_{w} (x_{w} \otimes w).
\end{equation}

We claim that, when $\cX$ is a simplex, it is possible to decompose each $x_{w}$ into a combination of elements in $\learnvert$ and each $y_{v}$ into a combination of elements in $\optvert$ in such a way that these two decompositions agree. That is, we can write $x_{w} = \sum_{v \in \learnvert} \gamma_{v, w}v$ and $y_{v} = \sum_{w \in \optvert} \gamma'_{v, w}w$ so that, for each $(v, w) \in \learnvert \times\optvert$,

\begin{equation}\label{eq:equiv-simplex-1}
    \lambda_{v}\gamma'_{v, w} = \mu_{w}\gamma_{v, w},
\end{equation}

\noindent
and we can therefore write

$$\csp = \sum_{v \in \learnvert} \sum_{w\in\optvert} \zeta_{v, w}(v \otimes w)$$

\noindent
where $\zeta_{v, w}$ is the common value of the two sides of \eqref{eq:equiv-simplex-1}. But now, note that this gives rise to a normal-form CE $\cspV \in \Delta(\learnvert\times\optvert)$ via $\cspV(v, w) = \zeta_{v,w}$. In particular, by \eqref{eq:equiv-cspx}, we have that $\NormSwap_{X}(\cspV) = 0$, and by \eqref{eq:equiv-cspy}, we have that $\NormSwap_{Y}(\cspV) = 0$.

It remains to show that a common decomposition (i.e., of the form in \eqref{eq:equiv-simplex-1}) exists. To see this, note that since $\cX$ is a simplex, there is a unique way to write each $x_w$ as a convex combination $\sum_{v \in \learnvert} \gamma_{v, w}v$. We claim that if we then define $\gamma'_{v, w} = (\mu_{w}\gamma_{v, w})/\lambda_{v}$ (so to satisfy \eqref{eq:equiv-simplex-1}), it must be the case that $y_{v} = \sum_{w \in \optvert} \gamma'_{v, w}w$. 

To see this, for any $v \in \learnvert$, consider the bilinear function $\rho_{v}: \learnset \otimes \optset \rightarrow \Rset\optset$ defined via $\rho_{v}(v \otimes w) = w$ and $\rho_{v}(v' \otimes w) = 0$ for $v' \neq w$ (note that this function only exists since $\cX$ is a simplex). Applying $\rho_v$ to \eqref{eq:equiv-cspx}, we have that $\rho_{v}(\phi) = \lambda_{v}y_{v}$. On the other hand, applying $\rho_v$ to \eqref{eq:equiv-cspy} (after substituting in our above decomposition for $x$), we have that $\rho_{v}(\phi) = \sum_{w}\mu_{w}\gamma_{v,w}w$. Equating these two expressions for $\rho_{v}(\phi)$, it follows that $y_{v} = \sum_{w} \gamma'_{v, w}w$, as desired.
\end{proof}

Finally, we conclude with a couple of remarks on the problem of actually computing equilibria. Note that our efficient no-profile-swap-regret learning algorithm (Theorem~\ref{thm:upper-semi-separation}) immediately gives us an efficient algorithm to compute \emph{some} profile CE in any (efficiently representable) polytope game $G$, by making both players run these low-profile-swap-regret dynamics. In particular, define an \emph{$\eps$-approximate profile CE} $\csp$ to be a CSP satisfying $\CorrDist(\csp) \le \varepsilon$ for both players. We have the following theorem.

\begin{theorem}\label{thm:ce-computation}
    Given efficient separation oracles for the sets $\learnset$ and $\optset$, there exists an algorithm that runs in polynomial time in $d_L, d_O, \frac{1}{\varepsilon}$ and computes an $\eps$-approximate profile CE.
\end{theorem}

\begin{proof}
    Our algorithm simulates repeated game play between the two players, with both players employing Algorithm~\ref{algo:semisep-approach} for $T = \Theta(\eps^{-2})$ rounds. The algorithm runs in polynomial time and by Theorem~\ref{thm:upper-semi-separation}, directly guarantees that $\CorrDist(\csp) = O(\sqrt{T}) \leq 1/\eps$ (for an appropriate choice of the constant in $\Theta(\eps^{-2})$) for both players.
\end{proof}

One interesting consequence of Theorems~\ref{thm:equiv-simplex} and \ref{thm:ce-computation} is that it leads to decentralized dynamics for efficiently computing a CSP corresponding to a normal-form CE in any polytope game where the action set of one of the two players is a simplex (e.g., the class of Bayesian games mentioned previously). However, it is also possible to efficiently compute a succinct representation of such a normal-form CE directly, by running learning dynamics where the simplex player plays a no-swap-regret algorithm and the other player best responds every round \citep{zhang_personal_communication}.

Theorem~\ref{thm:ce-computation} allows us to compute a single profile CE. On the other hand, if we want to \emph{optimize} over the set of profile CE, we need to contend with the hardness result in Theorem~\ref{thm:hardness} -- in fact, by setting the utilities of one of the players to zero (so that they are indifferent between all outcomes), the set of profile CE becomes exactly the no-profile-swap-regret menu, which is provably hard to optimize over. This situation is reminiscent of the situation in \cite{papadimitriou2008computing} for computing correlated equilibria in succinct multiplayer games (although note that here this phenomenon occurs for games with only two players).

\subsection*{Acknowledgments}

YM was partially supported by the European Research Council (ERC) under the European Union’s Horizon 2020 research and innovation program (grant agreement No. 882396), by the Israel Science Foundation,  the Yandex Initiative for Machine Learning at Tel Aviv University and a grant from the Tel Aviv University Center for AI and Data Science (TAD). This work was done in part while ERA was visiting the Simons Institute for the Theory of Computing. NC was partially supported by the IBM PhD Fellowship. 

The authors would like to thank Gabriele Farina, Maxwell Fishelson, Noah Golowich, Brian Hu Zhang, and Junyao Zhao for helpful discussions about an earlier draft of this work.

\bibliography{references}

\appendix
\section{Alternate formulation of profile swap regret}\label{app:alternate-formulation}

In this appendix, we present an alternate formulation of profile swap regret with some advantages for parsing this as a ``swap regret'' notion and comparing it to other notions of swap regret (i.e., the definition below is used in the that profile swap regret is upper bounded by polytope swap regret in Theorem~\ref{thm:comp-prof-swap}). 

We begin by showing that in the original definition of profile swap regret, it suffices to consider decompositions (of the form \eqref{eq:csp-decomp}) where each $x_{(k)}$ is a vertex of $\learnset$.

\begin{lemma}\label{lem:decomp-equiv}
Let $\csp = \sum_{k=1}^{K} \lambda_{k} (x_{(k)} \otimes y_{(k)})$ be an arbitrary decomposition of the CSP $\csp$. Then there exists another decomposition $\csp = \sum_{k=1}^{K'} \lambda'_{k} (x'_{(k)} \otimes y'_{(k)})$ where each $x'_{(k)} \in \learnvert$, all $x'_{(k)}$ are distinct, and

\begin{equation}
\label{eq:decomp_equiv}
\sum_{k=1}^{K'} \lambda'_{k}\Reg(x'_{(k)}, y'_{(k)}) \leq \sum_{k=1}^{K} \lambda_{k}\Reg(x_{(k)}, y_{(k)}).
\end{equation}
\end{lemma}
\begin{proof}
Fix any $k \in [K]$, and decompose $x_{(k)}$ arbitrarily into a convex combination $\sum_{s=1}^{S_k}\gamma_{s, k}v_{s, k}$ of vertices $v_{s, k} \in \learnvert$. Note that since the function $\Reg(x, y)$ is convex in $x$ for any fixed $y$ (it is the maximum of a collection of linear functions), it follows that $\sum_{s=1}^{S_k}\gamma_{s,k}\Reg(v_{s, k}, y_{(k)}) \leq \Reg(x_{(k)}, y_{(k)})$. It follows that the decomposition $(\lambda'_k, x'_{(k)}, y'_{(k)})$ formed by decomposing each $x_{(k)}$ in this way and aggregating the terms with the same vertex $v$ has the desired property. 
\end{proof}

One consequence of Lemma~\ref{lem:decomp-equiv} is that it allows us to write profile swap regret as a ``swap regret'': that is, regret with respect to the set of swap functions $\pi: \learnvert \rightarrow \learnvert$ (swapping vertices of $\learnset$), something not immediately clear from the original definition. In particular, we can equivalently define the profile swap regret $\CorrSwap(\csp)$ of a CSP via the following steps:

\begin{enumerate}
    \item Decompose $\csp$ into a convex combination of strategy profiles of the form

    \begin{equation}\label{eq:vertex-decomp}
    \csp = \sum_{v \in \learnvert} \lambda_{v} (v \otimes y_{v}),
    \end{equation}

    \noindent
    where $\lambda_v \geq 0$, $\sum_{v}\lambda_{v} = 1$ and $y_{v} \in \optset$.

    \item Define $\CorrSwap(\csp)$ to equal

    \begin{equation}\label{eq:vertex-csp}
    \CorrSwap(\csp) = T \cdot \left(\min_{y_{v}, \lambda_{v}} \max_{\pi^{*}:\learnvert \rightarrow \learnvert} \sum_{v \in \learnvert} \lambda_{v} \left(u_L(\pi^{*}(v), y_{v}) - u_L(v, y_{v})\right)\right),
    \end{equation}

    \noindent
    where the outer minimum is over all valid decompositions of the form \eqref{eq:vertex-decomp}. This already has a form similar to the definition of polytope swap regret (where we choose the best decomposition that minimizes our swap regret with respect to all functions mapping $\learnvert$ to $\learnvert$).

    \item Finally, if we desire, we can apply the minimax theorem to switch the order of the minimum and maximum in \eqref{eq:vertex-csp}. Doing so requires convexifying the domain over which we take the maximum, and therefore considering swap functions $\pi^{*}$ which send $\learnvert$ to $\learnset$ (instead of the ``pure'' swap functions which send $\learnvert$ to $\learnvert$).

    \begin{equation}\label{eq:vertex-csp-minimax}
    \CorrSwap(\csp) = T \cdot \left(\max_{\pi^{*}:\learnvert \rightarrow \learnset} \min_{y_{v}, \lambda_{v}}  \sum_{v \in \learnvert} \lambda_{v} \left(u_L(\pi^{*}(v), y_{v}) - u_L(v, y_{v})\right)\right),
    \end{equation}
\end{enumerate}

\section{Orthant-approachability via semi-separation}\label{app:semi-sep}

The goal of this appendix is to prove Theorem~\ref{thm:approach-semisep}; that access to an efficient semi-separation oracle suffices for efficiently performing Blackwell approachability. As mentioned in Section \ref{sec:algorithms}, the proof of this theorem follows very closely from the proof of \cite{daskalakis2024efficient}, who proved this theorem for the special class of approachability problems arising from linear swap-regret minimization, but in a fairly general way that easily extends to the setting of Theorem~\ref{thm:approach-semisep}.

To avoid the unnecessary redundancy of copying the entire proof of \cite{daskalakis2024efficient} with minor changes (but still in an attempt to be relatively self-contained here), we will cite two main ingredients from \cite{daskalakis2024efficient} (describing whatever changes, if any, need to be made to the proof to extend them to our setting) and then describe how to combine them into an efficient algorithm for orthant-approachability.

The first ingredient is a procedure \cite{daskalakis2024efficient} call \emph{Shell Projection}. Shell Projection can be thought of as a ``semi-separation'' variant of the operation of projecting an arbitrary $u$ to the set $\cU$ (a common primitive in many learning algorithms). In Shell Projection, instead of returning a projection of $u$ onto $\cU$ (which is computationally difficult without an actual separation oracle for $\cU$), we instead return a projection of $u$ onto some convex superset $\widetilde{\cU}$ of $\cU$, but with the guarantee that this projection is response-satisfiable.

Before we present this theorem, recall that a semi-separation oracle for a set $\cU$ of bi-affine functions takes any bi-affine function and either returns that $u$ is response-satisfiable or produces a hyperplane separating $u$ from $\cU$. Throughout the rest of this section, we will let $\cU_{RS}$ denote the (non-convex) set of response-satisfiable bi-affine functions $u$ (i.e., all $u$ with the property that there exists some $x \in \learnset$ such that $u(x, y) \leq 0$ for all $y \in \optset$). It is worth mentioning that nothing about the Shell Projection routine below (or the one originally presented in \cite{daskalakis2024efficient}) requires anything specific about the structure of response-satisfiable bi-affine functions -- indeed, the same theorem holds for a general notion of semi-separation of a convex set $\cK$ relative to a (possibly non-convex) superset $\cK_{NC} \supseteq \cK$, where given a point $x$ in the ambient space of $\cK$, the semi-separation oracle either returns that $x$ belongs to $\cK_{NC}$ or returns a separating hyperplane separating $x$ from $\cK$.

\begin{theorem}[Shell Projection]\label{thm:dffps-proj}
    There is an algorithm $\ShellProj$ which takes as input:

    \begin{itemize}
        \item A semi-separation oracle separating a convex set $\cU$ relative to some set $\cU_{RS} \supset \cU$,
        \item A known ball of radius $\rho$ contained within $\cU$,
        \item A superset $\cU'$ of $\cU$ with $\diam(\cU') \leq D$, provided via an efficient separation oracle,
        \item An element $u \in \cU'$, and
        \item A precision $\eps > 0$,
    \end{itemize}

    \noindent
    runs in time $\poly(\dim(\cU), \eps^{-1}, \rho^{-1}, D)$, and returns:
    \begin{itemize}
        \item A convex set $\widetilde{\cU}$ satisfying $\cU \subseteq \widetilde{\cU} \subseteq \cU'$, constructed by intersecting $\cU'$ with at most $\poly(\dim(\cU), \eps^{-1}, \rho^{-1}, D)$ half-spaces, and
        \item A \emph{response-satisfiable} $\tilde{u} \in \widetilde{\cU}$ satisfying

        $$||\tilde{u} - \Proj_{\widetilde{\cU}}(u)|| \leq \eps.$$
    \end{itemize}
\end{theorem}
\begin{proof}
The proof directly follows from the proof of Theorem 4.4 in \citep{daskalakis2024efficient}, with the following changes:

\begin{itemize}
    \item $\Phi(\cP)$, $\Phi_{FP}$, $\tilde{\Phi}$, $\cM$, $\phi$, and $\phi'$ should be replaced by $\cU$, $\cU_{RS}$, $\widetilde{\cU}$, $\cU'$, $u$, and $u'$ respectively. 
    \item Theorem 4.4 uses the fact that if $\cP$ circumscribes and is inscribed in balls of radii $r$ and $R$ respectively, then $\Phi(\cP)$ contains a ball of radius $r/2R$. Here, we simply make the assumption that $\cU$ contains a ball of radius $\rho$ (and so, $\rho^{-1}$ appears in place of $R/r$ in our time complexity).
    \item The dimension of $\Phi(\cP)$ is $d(d+1)$ (where $d = \dim(\cP)$), and a polynomial dependence in $d$ appears in the time complexity. Instead, we directly include a polynomial dependence on $\dim(\cU)$ in our time complexity.
\end{itemize}

At a very high-level, this Shell Projection routine works by first strengthening the semi-separation oracle into a subroutine called Shell Ellipsoid, which can semi-separates \emph{convex sets} $\cK$ from $\cU$ instead of individual points (returning either a hyperplane separating $\cK$ from $\cU$, or a point in $\cK$ belonging to $\cU_{RS}$). They then repeatedly run Shell Ellipsoid on larger and larger balls centered at $u$ until they find one with a point $\tilde{u}$ in $\cU_{RS}$ -- by the guarantees of Shell Ellipsoid, they can then be guaranteed that the point $\tilde{u}$ is close to a projection of $u$ onto some set containing $\cU$.
\end{proof}

The second ingredient we will need from \cite{daskalakis2024efficient} is a slight variant of projected gradient descent they call Shell Gradient Descent (Algorithm 2 in Section 4.3 of \cite{daskalakis2024efficient}). Shell Gradient Descent solves the following variant of standard online linear optimization (that we will call Shell OLO). In Shell OLO, at the beginning of every round $t$, the learner is told (i.e., given efficient oracle access to) a ``shell'' set $\cX_{t}$, from which they must pick their action $x_{t}$ in round $t$. Their goal is to compete with the best fixed action in (some subset of) the intersection of all the sets $\cX_{t}$. \cite{daskalakis2024efficient} prove the following theorem.

\begin{theorem}[Shell Gradient Descent]\label{thm:dffps-grad}
Let $\cX_{1}, \cX_{2}, \dots, \cX_{T}$ be an arbitrary sequence of convex sets satisfying $\cX_{t} \subseteq \cB_{d}(0, D)$, let $\ell_{1}, \ell_{2}, \dots, \ell_{T} \in [-1, 1]^d$ be an arbitrary sequence of adversarial losses, and let $\cX$ be an arbitrary subset of $\bigcap_{t=1}^{T} \cX_{t}$. Then, if we choose $x_1 \in \cX_1$ arbitrarily and choose $x_{t} = \Proj_{\cX_t}(x_{t-1} - \eta_{t-1}\ell_{t-1})$ (for some sequence of step sizes $\eta_t > 0$), this sequence of actions satisfies

$$\max_{x^* \in \cX} \sum_{t=1}^{T} \langle x_t - x^{*}, \ell_t\rangle \leq \frac{D^2}{2\eta_{T}} + \sum_{t=1}^{T}\frac{\eta_{t}}{2} ||\ell_{t}||^2$$
\end{theorem}
\begin{proof}
See Theorem 4.3 of \cite{daskalakis2024efficient}. (Note that this follows nearly immediately from the standard analysis of projected gradient descent).
\end{proof}

We now describe how these pieces fit together to produce an efficient algorithm for orthant-approachability (and thus, allow us to prove Theorem~\ref{thm:approach-semisep}). The starting point is the main idea behind the standard reduction from approachability to online linear optimization (popularized by \cite{abernethy2011blackwell}), which is that if we can choose a sequence of bi-affine functions $u_t \in \cU$ that forces the quantity 

$$\Reg^{\mathrm{dual}}(\bu, \bx, \by) = \max_{u^{*} \in \cU} \sum_{t=1}^{T} u^{*}(x_t, y_t) - \sum_{t=1}^{T} u_{t}(x_t, y_t)$$

\noindent
to grow sublinearly in $T$, and in addition choose our actions $x_t \in \learnset$ so that $u_{t}(x_t, y) \leq 0$ for any $y \in \optset$, it follows that the quantity

$$\AppLoss(\bx, \by) = \max_{u^{*} \in \cU}\sum_{t=1}^{T} u^{*}(x_t, y_t)$$

\noindent
is at most $\Reg^{\mathrm{dual}}(\bu, \bx, \by)$ and is therefore sublinear in $T$. 

Ordinarily, if we are given oracle access to $\cU$, we can guarantee that $\Reg^{\mathrm{dual}}(\bu, \bx, \by) = o(T)$ by simply running any no-external-regret algorithm for online linear optimization where the action set is $\cU$ and the reward function in round $t$ is the linear function sending $u$ to $u(x_t, y_t)$. But without oracle access to $\cU$, it is hard to even guarantee that $u_t$ belongs to $\cU$. The key observation of \cite{daskalakis2024efficient} is that it is not necessary that $u_t$ belong $\cU$, but just that each $u_t$ is response-satisfiable (so that $u_t(x_t, y_t) \leq 0$ holds for some $x_t \in \learnset$ independently of $y_t$) and that the sequence of $u_t$ incurs low regret.

This is where Theorems~\ref{thm:dffps-proj} and \ref{thm:dffps-grad} enter the picture. We will choose our sequence $u_t$ by running Shell Gradient Descent over sets generated by $\ShellProj$. This gives us the regret guarantees of Shell Gradient Descent (Theorem~\ref{thm:dffps-grad}) while also ensuring that each $u_t$ is response-satisfiable (by Theorem~\ref{thm:dffps-proj}).

\begin{algorithm}[ht]
    \caption{Orthant-approachability algorithm via semi-separation (analogue of Algorithm 
4 in \cite{daskalakis2024efficient}) 
}\label{algo:semisep-approach}
\raggedright
    
    \textbf{Input:} a superset $\cU'$ of $\cU$ contained within $\cB_{d}(D)$, a ball of radius $\rho$ contained within $\cU$, and a semi-separation oracle separating $\cU$ from $\cU_{RS}$. 

    Set step size $\eta := \frac{D}{D_{\cX}D_{\cY}}T^{-1/2}$ and precision $\epsilon := \frac{D}{D_{\cX}D_{\cY}}T^{-1/2}$
    
    Let $u_1$ be any point in $\cU'$ and $x_1$ be any point in $\learnset$

    \For{$t = 1, 2, \dots, T$}{
        Output $x_t \in \learnset$ and receive optimizer action $y_t \in \optset$.

        Set $L_t \in \Rset^{\dim(\cU)}$ so that for any element $u \in \cU$, $\langle u, L_t \rangle = -u(x_t, y_t)$.

        Run $\ShellProj$ on $u_t - \eta L_t$ with superset $\cU'$ and precision $\epsilon$, receiving a shell set $\widetilde{\cU}_{t+1}$ (only used in the regret analysis) and response-satisfiable element $u_{t+1} \in \widetilde{\cU}_{t+1} \cap \cU_{RS}$.

        Compute an $x_{t+1} \in \learnset$ such that $u_{t+1}(x_{t+1}, y) \leq 0$ for all $y \in \optset$ (guaranteed to exist since $u$ is response-satisfiable, and computable efficiently via linear programming).
    }
\end{algorithm}

\begin{theorem}[Approachability via a semi-separation oracle (restatement of Theorem~\ref{thm:approach-semisep})]\label{thm:approach-semisep-restate}
Consider an orthant-approachability instance $(\cX, \cY, \cU)$ where $\cX \subseteq \cB(D_{\cX})$, $\cY \subseteq \cB(D_{\cY})$, and $\cB(u_{0}, \rho) \subseteq \cU \subseteq \cB(D)$ (for some known $u_{0}$ and radius $\rho$). If we have access to efficient separation oracles for $\cX$ and $\cY$, and an efficient semi-separation oracle for the set $\cU$, then Algorithm~\ref{algo:semisep-approach} has the guarantee that

$$\max_{u \in \cU} \frac{1}{T}\sum_{t=1}^{T} u\left(x_t, y_t \right) \leq O(D_{\cX}D_{\cY}D\sqrt{T}),$$

\noindent
and runs in time $\poly(D_{\cX}, D_{\cY}, D, \rho^{-1}, \dim(\cU), T)$.
\end{theorem}
\begin{proof}
This proof closely follows the proof of Theorem 4.5 in \cite{daskalakis2024efficient}.

Note that the sequence $u_t$ generated by Algorithm~\ref{algo:semisep-approach} is almost the same sequence as would be generated by running Shell Gradient Descent on the sequence of shell sets $\widetilde{\cU}_{t}$, with the only caveat that $u_{t+1}$ is distance at most $\eps$ from the true projection of $u_{t} - \eta L_{t}$ onto $\widetilde{\cU}_{t+1}$. The guarantees of Theorem~\ref{thm:dffps-grad} then imply that

\begin{equation}\label{eq:app-semisep-proof-1}
    \max_{u^* \in \cU} \sum_{t=1}^{T} \langle u_t - u^{*}, L_t\rangle \leq \frac{D^2}{2\eta} + \sum_{t=1}^{T}\frac{\eta}{2} ||L_{t}||^2 + \eps\sum_{t=1}^{T} ||L_{t}|| = O(DD_{\cX}D_{\cY}\sqrt{T}), 
\end{equation}

\noindent
where here we have used the fact that $||L_{t}|| \leq D_{\cX}D_{\cY}$ and $\eta = \eps = \frac{D}{D_{\cX}D_{\cY}}T^{-1/2}$. But since $\langle u_t, L_t \rangle = -u_{t}(x_t, y_t)$, the left hand side of \eqref{eq:app-semisep-proof-1} is equal to $\max_{u^{*} \in \cU} \sum_{t=1}^{T} u^{*}(x_t, y_t) - \sum_{t=1}^{T} u_t(x_t, y_t) = \Reg^{\mathrm{dual}}(\bu, \bx, \by)$. Moreover, since each $u_{t}(x_t, y_t) \leq 0$ (by the choice of $x_t$), we have that 

$$\max_{u \in \cU} \frac{1}{T}\sum_{t=1}^{T} u\left(x_t, y_t \right) \leq \Reg^{\mathrm{dual}}(\bu, \bx, \by) \leq O(DD_{\cX}D_{\cY}\sqrt{T}),$$

\noindent
as desired.
\end{proof}
\section{Additional results on game-agnostic learning}\label{app:game-agnostic}

\subsection{Profile swap regret is not computable from the game agnostic transcript}\label{app:prof-from-transcript}

In this appendix, we show that it is not possible to compute the profile swap regret $\CorrSwap(\bx, \by)$ from the game-agnostic transcript $(\bx, \br)$ corresponding to the game-dependent transcript $(\bx, \by)$. Specifically, we have the following theorem.

\begin{theorem}\label{thm:prof-from-transcript}
There exists a polytope game $G$, and two transcripts $(\bx_1, \by_1)$ and $(\bx_2, \by_2)$ that have the same game-agnostic transcript $(\bx, \br)$, but where $\CorrSwap(\bx_1, \by_1) \neq \CorrSwap(\bx_2, \by_2)$.
\end{theorem}
\begin{proof}
Our polytope game $G$ will be a variant of the polytope game used to separate profile swap regret and polytope swap regret in the proof of Theorem~\ref{thm:comp-prof-swap}. Specifically, we will consider the game where $\learnset = \Delta_{2}^2$ and $\optset = \Delta_{2}^2 \times \Delta_{4}$. We will write elements $y \in \optset$ in the form $y = (y_r, y_s)$ with $y_r \in \Delta_{2}^2$ and $y_s \in \Delta_4$, and write $\Proj_1(y_{r}, y_{s}) = y_r \in \Delta_{2}^{2}$ and $\Proj_2(y_{r}, y_s) = y_{s}\in \Delta_4$. Finally, we define $u_L(x, y) = \langle x, \Proj_{1}(y)\rangle$.

As in the proof of Theorem~\ref{thm:comp-prof-swap}, we will write $\learnvert = \{v_{11}, v_{12}, v_{21}, v_{22}\}$, where $v_{11} = (1, 0, 1, 0)$, $v_{12} = (1, 0, 0, 1)$, $v_{21} = (0, 1, 1, 0)$, and $v_{22} = (0, 1, 0, 1)$. We will likewise write $\optvert = \{w_{ijk}\}_{i, j \in [2], k \in [4]}$, where $w_{ijk} = (v_{ij}, e_k)$. In particular, note that $\Proj_{1}(w_{ijk}) = v_{ij}$.

We now define two transcripts of play. Each transcript of play will be divided into four epochs of $T/4$ rounds during which both players play a fixed action. In the first transcript $(\bx_1, \by_1)$, these four action pairs are $((v_{11}, w_{111}), (v_{12}, w_{121}), (v_{21}, w_{211}), (v_{22}, w_{111}))$. In the second transcript $(\bx_2, \by_2)$, these four action pairs are $((v_{11}, w_{111}), (v_{12}, w_{122}), (v_{21}, w_{213}), (v_{22}, w_{114}))$. Note that since $\bx_{1, t} = \bx_{2, t}$ and $\Proj_{1}(\by_{1, t}) = \Proj_{1}(\by_{2, t})$ for each $t \in [T]$, it follows that both transcripts have the same game-agnostic transcript $(\bx, \br)$. 

Let $\csp_1$ be the CSP corresponding to the first transcript and $\csp_2$ be the CSP corresponding to the second transcript. Note that by following the same derivation as in the proof of Theorem~\ref{thm:comp-prof-swap}, we can write

$$\csp_1 = \frac{1}{4}(v_{12} \otimes (w_{121} + w_{111}) + v_{21} \otimes (w_{211} + w_{111})),$$

\noindent
from which we can conclude that $\CorrSwap(\csp_1) = 0$. 

On the other hand, we claim that the only decomposition of $\csp_2$ into the form

\begin{equation}\label{eq:app-prof-noncomp-1}
\csp_2 = \sum_{v \in \learnvert} \lambda_{v} (v \otimes y_{v}),.
\end{equation}

\noindent
is the decomposition given by its construction as a CSP, namely,

\begin{equation}\label{eq:app-prof-noncomp-2}
\csp_2 = \frac{1}{4}\left(v_{11}\otimes w_{111} + v_{12}\otimes w_{122} + v_{21} \otimes w_{213} + v_{22} \otimes w_{114}\right).
\end{equation}

To see this, for any $k \in [4]$, consider the bilinear function $\rho_k$ defined via $\rho_k(x \otimes y) = \langle \Proj_{2}(y), e_k\rangle x$. For each $k$, there is exactly one $v_k \in \learnvert$ and $w_{k} \in \optvert$ with the property that $v_{k} \otimes w_{k}$ is the only term in \eqref{eq:app-prof-noncomp-2} that does not vanish under $\rho_k$ (e.g., for $k=4$ it is $v_{22} \otimes w_{114}$). But applying $\rho_k$ to \eqref{eq:app-prof-noncomp-1},

$$\rho_k(\csp_2) = \sum_{v \in \learnvert} \lambda_{v}\langle y_v, e_k\rangle v.$$

Since $v_k$ is an extreme point of $\learnset$, this implies that $\langle y_{v}, e_{k}\rangle$ can be non-zero only for $v = v_k$ (for which value it must equal $1$). It follows that the only possible decomposition of $\csp_2$ is that given in \eqref{eq:app-prof-noncomp-2}. But then, since $v_{22} \not\in \BR_{L}(w_{114}) = \{v_{11}\}$, it follows that $\CorrSwap(\csp_2) > 0$ and therefore $\CorrSwap(\csp_1) \neq \CorrSwap(\csp_2)$.

\end{proof}

\subsection{Game-agnostic profile swap regret does not upper bound profile swap regret}\label{app:game-agnostic-profile-bad}

Here we prove Theorem~\ref{thm:game-agnostic-profile-bad}, showing that the game-agnostic profileswap regret $\CorrSwap_{G^{*}}(\bx, \br)$ of a transcript does not necessarily bound its (standard, game-aware) profile swap regret.

\begin{proof}[Proof of Theorem~\ref{thm:game-agnostic-profile-bad}]
We build off of the example in the proof of Theorem~\ref{thm:prof-from-transcript}. Consider specifically the second transcript $(\bx_2, \by_2)$ where $\CorrSwap_{G}(\bx_2, \by_2) > 0$. 

Note that if we define (for any $i, j \in [2]$) $r_{ij} \in \learnset^{*}$ so that $r_{ij}(x) = \langle v_{ij}, x \rangle$, then $\by_{2}$ corresponds to the sequence of rewards $\br_{2} = (r_{11}, r_{12}, r_{21}, r_{11})$ . We now claim that $\CorrSwap_{G^{*}}(\bx_2, \br_2) = 0$. This will follow for similar logic to the original proof that $\CorrSwap_{G}(\bx_1, \by_1) = 0$. Indeed, the CSP $\csp^{*}$ of $(\bx_{2}, \br_{2})$ in $G^{*}$ is given by

$$\phi^{*} = \frac{1}{4}(v_{11} \otimes r_{11} + v_{12} \otimes r_{12} + v_{21} \otimes r_{21} + v_{22} \otimes r_{11})$$

By following the same derivation as in the proof of Theorem~\ref{thm:comp-prof-swap}, we can rewrite this in the form

$$\phi^{*} = \frac{1}{4}(v_{12} \otimes (r_{11} + r_{12}) + v_{21} \otimes (r_{11} + r_{21})).$$

We can check that this decomposition has zero profile swap regret under $G^*$, and so $\CorrSwap_{G^*}(\bx_2, \br_2) = 0$.
\end{proof}
\section{Additional results on correlated equilibria}\label{app:correlated-equilibria}

\subsection{Mediator protocols correspond to normal-form CE}\label{app:mediator-nf}

In this appendix we make clear the connection between the incentive-compatible signaling schemes a mediator can implement third-party mediator and normal-form correlated equilibria. In particular, we show that there is no need for a mediator to ever send recommendations to the two players that are not extreme points of $\learnvert$ or $\optvert$, and hence that the set of valid mediator outcomes can be exactly characterized by normal-form CE.

Specifically, for any (finite support) distribution $\sigma \in \Delta(\learnset \times \optset)$, we can consider the signaling scheme where a mediator samples a strategy profile $(x, y) \sim \sigma$, privately recommends playing $x$ to the $\cX$-player, and privately recommends playing $y$ to the $\cY$-player. We say this signaling scheme is incentive compatible if neither player has incentive to deviate given their observation. That is, for any $x_{\mathrm{sig}} \in \cX$, there should not exist an $x' \in \cX$ such that:

$$\E_{(x, y) \sim \sigma}[u_X(x', y) \mid x = x_{\mathrm{sig}}] > \E_{(x, y) \sim \sigma}[u_X(x, y) \mid x = x_{\mathrm{sig}}].$$

(i.e., the $\cX$-player cannot increase their utility by playing $x'$ every time they receive the signal $x_{\mathrm{sig}}$). Similarly, for any  $y_{\mathrm{sig}} \in \cY$, there should not exist a $y' \in \cY$ such that:

$$\E_{(x, y) \sim \sigma}[u_Y(x, y') \mid y = y_{\mathrm{sig}}] > \E_{(x, y) \sim \sigma}[u_Y(x, y) \mid y = y_{\mathrm{sig}}].$$

The following lemma shows that we can always convert an incentive-compatible signaling scheme supported on the entire set $\learnset \times \optset$ to one supported on pairs of pure strategies $\learnvert \times \optvert$ with the same utility for both players (in fact, that induce the same CSP).

\begin{lemma}
Let $\sigma \in \Delta(\learnset \times \optset)$ be an incentive compatible signaling scheme. Let $\sigma' \in \Delta(\learnvert \times \optvert)$ be the signaling scheme formed by first sampling a strategy profile $(x, y)$ from $\sigma$, decomposing $x$ and $y$ into convex combinations of pure strategies $x = \sum_{v \in \learnvert} \lambda_{v}v$ and $y = \sum_{w \in \optvert}\mu_{w}w$, and then returning the pure strategy profile $(v, w)$ with probability $\lambda_{v}\mu_{w}$. Then $\sigma'$ is an incentive compatible signaling scheme, and $\E_{(x, y) \sim \sigma}[x \otimes y] = \E_{(v, w) \sim \sigma'}[v \otimes w]$.
\end{lemma}
\begin{proof}

The fact that $\E_{(x, y) \sim \sigma}[x \otimes y] = \E_{(v, w) \sim \sigma'}[v \otimes w]$ follows directly from the construction of $\sigma'$ (in particular, $\sum_{v, w}\lambda_{v}\mu_{w}(v \otimes w) = \left(\sum_{v} \lambda_{v}v\right) \otimes \left(\sum_{w} \mu_{w}w\right) = x \otimes y$). It therefore suffices to show that $\sigma'$ is incentive compatible.

We will do this for the $\cX$-player (it follows for the $\cY$-player by symmetry). Consider any $x \in \learnset$ which is recommended to the $\learnset$-player with positive probability under $\sigma$. Let $y(x) \in \optset$ be the expected action recommended to the $\optset$-player, conditioned on the fact that $x$ is recommended to the $\learnset$-player. Note that since $\sigma$ is incentive compatible, this means that it must be the case that $x \in \BR_{X}(y(x))$ (or the $\learnset$ player could increase their utility by deviating from $x$ to any element of $\BR_{X}(y(x))$).

Now, consider any $v^* \in \learnvert$ such that $\lambda_{v^*} > 0$ in some decomposition of $x$ of the form $x = \sum_{v \in \learnvert} \lambda_{v}v$. We claim that it must also be the case that $v^* \in \BR_{X}(y(x))$. If not, the $\cX$-player could again improve their utility by deviating from $x$ to the strategy formed by replacing the $\lambda_{v^*}v^*$ term in the decomposition $\sum_{v \in \learnvert} \lambda_{v}v$ of $x$ with $\lambda_{v^*}v'$ for some $v' \in \BR_{X}(y(x))$. But now, if we let $y(v^{*}) \in \optset$ to be the expected message in $\sigma'$ received by the $\optset$-player conditioned on the $\learnset$-player receiving $v^*$, $y(v^{*})$ is a convex combination of elements of the form $y(x)$ for $x$ with $\lambda_{v^*} > 0$. It follows that $v^{*} \in \BR_{X}(y(v^{*}))$, and therefore that $\sigma'$ is incentive-compatible.
\end{proof}

\subsection{Every profile CE can be reached by no-profile-swap-regret dynamics}\label{app:all-profile-ce}

Just as normal-form CE exactly characterize the set of outcomes implementable by an incentive-compatible signaling scheme, profile CE exactly characterize the set of outcomes that can occur as a result of no-profile-swap-regret dynamics. One direction of this characterization (that such dynamics result in profile CE) is immediate from the definition of profile CE -- below we show that every profile CE can be realized by such dynamics, completing this characterization.

\begin{lemma}
    Given any profile CE $\phi$, there exist profile swap regret minimizing algorithms for the two players that converge to this CSP when employed against each other.
\end{lemma}

\begin{proof}
    Let $\phi = \sum_{k=1}^{K} \lambda_{k} (x_{(k)} \otimes y_{(k)})$ be a decomposition of the CSP $\phi$ with $K \le d_{\learnset}$ (where $\learnset \subseteq \mathbb{R}^{d_{\learnset}}$) -- Lemma~\ref{lem:decomp-equiv} guarantees such a decomposition. In case any of the $\lambda_k$'s are not rational -- we associate it with a sequence  $\lambda_k^1, \lambda_k^2 \cdots $ of positive rational numbers converging to $\lambda_k$. The two algorithms attempt to follow the schedule described below with any deviation observed triggering a default to a standard profile swap regret minimizing algorithm (such as the algorithm guaranteed by Theorem~\ref{thm:upper-approachability}). The schedule consists of an infinite series of epochs, the $i$-th epoch consists of $\sum_{k=1}^K \lambda^i_k$ rounds -- in each of the first $\lambda^i_1$ rounds in this epoch, the two players play $x_1$ and $y_1$ respectively; in the next $\lambda^i_2$ rounds in this epoch, the two players play $x_2$ and $y_2$ respectively and so on. It it not hard to see that both players following the schedule leads to convergence to the CSP $\phi$, which is a profile CE. On the other hand, any deviation from the schedule by an opposing player results in a reset to a profile swap regret minimizing algorithm, so both player's algorithms retain the worst case profile swap regret minimizing property.
\end{proof}

\subsection{One-sided mediator protocols and profile swap regret}\label{app:one-sided}

We can also think of profile CE as outcomes implementable by ``one-sided mediator'' protocols: in particular, a CSP $\csp$ is a profile CE if each player can imagine a mediator protocol implementing $\csp$ that is incentive compatible for themselves (but note that these mediator protocols may differ across players, and need not be incentive-compatible for other players).

To prove this, it suffices to show that any CSP $\csp$ with low profile swap regret for a specific player (e.g., the $\cX$-player) can be implemented by a mediator protocol that is incentive compatible for that player (and in particular, corresponds to a vertex-game CSP $\cspV$ with the property that $\NormSwap_{\cX}(\cspV) = 0$). We prove this below.

\begin{lemma}\label{lem:one-sided-mediator}
Let $\csp \in \learnset \otimes \optset$ be a CSP with the property that $\CorrSwap_{X}(\csp) = 0$. Then there exists a vertex-game CSP $\cspV \in \Delta(\learnvert \times \optvert)$ with the property that $\NormSwap_{\cX}(\cspV) = 0$ and $\Proj(\cspV) = \csp$.
\end{lemma}
\begin{proof}
Consider any normal-form swap regret minimizing algorithm $\cA$. Consider the menu of $\cA$ 


Let $G$ be the polytope game at hand. Consider a no-normal-form-swap-regret algorithm $\mathcal{A}$ in the vertex game $G^{V}$ (for the $\cX$-player). We can run $\mathcal{A}$ in the original polytope game by arbitrarily decomposing every adversary action $y_t \in \optset$ we see into a distribution $\yV_{t} \in \Delta(\optvert)$ (with $\E[\yV_{t}] = y_t$), obtaining the vertex game response $\xV_{t} \in \Delta(\learnvert)$ to it, and playing $x_t = \E[\xV_{t}]$ in the polytope game.

Note that since profile swap regret is upper bounded by normal-form swap regret (Theorems~\ref{thm:swap-notions-ordering} and \ref{thm:comp-prof-swap}), $\cA$ is also a no-profile-swap-regret algorithm in $G$. By Theorem~\ref{thm:poly_minimal}, this means that the asymptotic menu $\cM(\cA)$ is the no-profile-swap-regret menu $\cM_{NPSR}$, and in particular, any $\csp \in \cM_{NPSR}$ (i.e., any $\csp$ with $\CorrSwap_{X}(\csp) = 0$) can be asymptotically realized by sequence of opponent actions. But then, the corresponding vertex CSP $\cspV$ (realized by the limit of the transcript of play $(\bxV, \byV)$ in the vertex game) must satisfy $\NormSwap_{X}(\cspV) = 0$ (since $\cA$ is also no-normal-form-swap-regret) and $\Proj(\cspV) = \csp$. The conclusion follows.
\end{proof}

One consequence of Lemma~\ref{lem:one-sided-mediator} is that the no-profile-swap-regret menu is the projection of the no-swap-regret menu of the vertex game: that is, $\cM_{NPSR} = \Proj(\cM^{V}_{NSR})$, where $\cM^{V}_{NSR}$ is the set of vertex-game CSPs $\cspV$ with $\NormSwap(\cspV) = 0$. 


\section{Omitted proofs}\label{sec:omitted}

\subsection{Proof of Theorem~\ref{thm:nfg-swap-notions}}

\begin{proof}[Proof of Theorem~\ref{thm:nfg-swap-notions}]
When $\learnset = \Delta_{m}$, the set of affine linear transformations that sends $\learnset$ to itself can be expressed the set of $m$-by-$m$ row-stochastic matrices (this follows since every unit vector must map to an element of $\Delta_{m}$). This is a convex set whose extreme points are given by row-stochastic matrices which correspond to swap functions; i.e., 0/1 row-stochastic matrices where each row is a unit vector. Maximizing over all transformations in this set is equivalent to maximizing over all swap functions, so $\LinSwap(\bx, \by) = \Swap(\bx, \by)$.

For polytope swap regret, note that when $\learnset = \Delta_{m}$, there is a unique way to decompose an element in $x \in \learnset$ into an element $\xV \in \Delta(\learnvert)$ (specifically, $x$ itself provides the unique such decomposition). This means that we can simplify the definition of polytope swap regret in this case to $\PolySwap(\bx, \by) =  \max_{\pi: \learnvert \rightarrow \learnvert}(\sum_{t=1}^{T} u_L(\pi(x_t), y_t) - u_L(x_t, y_t))$, which is identical to the definition of $\Swap(\bx, \by)$.

Finally, again, the fact that any action $\xV \in \Delta(\learnvert)$ with $\E_{v \sim \xV}[v] = x$ must in fact satisfy $\xV_i = x_i$ for all $i \in [m]$ means that for normal-form games, $\NormSwap(\bx, \by) = \Swap(\bx, \by)$.
\end{proof}

\subsection{Proof of Theorem~\ref{thm:swap-notions-ordering}}

\begin{proof}[Proof of Theorem~\ref{thm:swap-notions-ordering}]
We begin by proving the two inequalities above. To show that $\LinSwap(\bx, \by) \leq \PolySwap(\bx, \by)$, first note that we can relax the set of swap functions $\pi$ we consider in the definition of polytope swap regret to all functions $\pi$ from $\learnvert$ to $\learnset$ (instead of just functions from $\learnvert$ to $\learnvert$). This relaxation does not change the value of polytope swap regret, since for any vertex $v \in \learnvert$, the value of $\pi(v)$ that maximizes polytope swap regret will be achieved at an extreme point of $\learnvert$. But now, for any affine linear function $\psi: \learnset \rightarrow \learnset$, if we consider the swap function $\pi$ with $\pi(v) = \psi(v)$ for all $v \in \learnset$, it is the case that $\overline{\pi}(\xV_t) = \E_{v \sim \xV_t}[\psi(v)] = \psi(\E_{v\sim \xV_{t}}[v]) = \psi(x_t)$ (in particular, since $\psi$ is linear, it commutes with taking expectations). It follows that the polytope swap regret is at least $\sum_{t=1}^{T} u_L(\psi(x_t), y_t) - \sum_{t=1}^{T} u_L(x_t, y_t)$, for any linear function $\psi$ that maps $\learnset$ to itself. Since $\LinSwap(\bx, \by)$ is the maximum of this quantity over all such linear functions $\psi$, it follows that $\LinSwap(\bx, \by) \leq \PolySwap(\bx, \by)$.

To show that $\PolySwap(\bx, \by) \leq \NormSwap(\bxV, \byV)$, note that we can write $\PolySwap(\bx, \by)$ as the minimum of $\NormSwap(\bxV, \byV)$ over all choices of $\bxV$ and $\byV$ with $x_{t} = \E[\bxV_{t}]$ and $y_{t} = \E[\byV_{t}]$. Since the $\bxV$ and $\byV$ in this theorem also satisfy these constraints, it immediately follows that $\PolySwap(\bx, \by) \leq \NormSwap(\bxV, \byV)$.

We now provide examples of transcripts exhibiting the gaps described in the second part of the theorem. An example of a family of transcripts where $\LinSwap(\bx, \by) = 0$ but $\PolySwap(\bx, \by) = \Omega(T)$ can be found in Theorem 7 of \cite{MMSSbayesian}. 

To construct a family of transcripts where $\PolySwap(\bx, \by) = 0$ but $\NormSwap(\bxV, \byV) = \Omega(T)$, consider the polytope game where $\cX = [0, 1]^2$ (and therefore $V(\cX) = \{(0,0), (1, 0), (1, 1), (0, 1)\}$) and $\cY = [-1, 1]$, with $u_{L}((x_1, x_2), y) = y(x_1 - x_2)$. Consider the transcript $(\bx, \by)$ where $x_t = (1/2, 1/2)$ for all $t \in [T]$ but $y_t = -1$ for rounds $t \in [1, T/2]$ and $y_t = 1$ for rounds $t \in [T/2 + 1, T]$.

Note that if we decompose every $x_t$ into the distribution $\frac{1}{2}(0, 0) + \frac{1}{2}(1, 1) \in \Delta(V(\cX))$, there is no swap function $\pi^*: V(\cX) \rightarrow V(\cX)$ that will improve the utility of the learner -- for example, on average when the learner plays the vertex $(0,0)$, they face the average optimizer strategy of $\overline{y} = 0$, and all elements of $V(\cX)$ are best responses to $\overline{y}$. It follows that $\PolySwap(\bx, \by) = 0$. 

On the other hand, consider the vertex game transcript where $\xV_t = \frac{1}{2}(0, 0) + \frac{1}{2}(1, 1)$ for rounds $t \in [1, T/2]$, and $\xV_t = \frac{1}{2}(0, 1) + \frac{1}{2}(1, 0)$ for rounds $t \in [T/2 + 1, T]$ (with $\yV_t = y_t$ for all $t \in [T]$). It is straightforward to check that this transcript satisfies $\E[\xV_t] = x_t$ and $\E[\yV_t] = y_t$. However, on the rounds where the learner plays $(0, 0)$, they face an average optimizer strategy of $\overline{y} = -1$. Against this optimizer strategy, the unique best-response for the learner is the strategy $(0, 1)$, and the learner can increase their total utility by $T/4$ by applying this swap. It follows that $\NormSwap(\bxV, \byV) = \Omega(T)$.
\end{proof}

\subsection{Proof of Theorem~\ref{thm:menu_char}}

\begin{proof}[Proof of Theorem~\ref{thm:menu_char}]
See the proof of Theorem 3.3 in \cite{paretooptimal} (this proof is for a normal-form game $G$, but does not use anything specific about the space of CSPs, and applies essentially as is to our setting).

It is worth noting that one direction of this argument -- that $\cM$ needs to contain some CSP of the form $x \otimes y$ for each $y \in \optset$ follows simply from the fact that the optimizer can always play the action $y_t = y$, in which case some element of this form is guaranteed to exist in $\cM$. It is also possible to prove a slightly weaker version of the converse directly from Blackwell's approachability theorem (which is all we need in this paper): that if $\cM$ contains some CSP of the form $x \otimes y$ for each $y \in \optset$, then some \emph{subset} $\cM'$ of $\cM$ is a valid asymptotic menu (see Lemma 3.4 in \cite{paretooptimal}). Proving the full converse requires showing that we can always extend a menu $\cM'$ to any superset of $\cM'$ (which can be accomplished by offering the optimizer the choice to play a choice of CSPs outside $\cM'$ -- see Lemma 3.5 in \cite{paretooptimal}). 
\end{proof}

\subsection{Proof of Theorem~\ref{thm:comp-prof-swap}}

\begin{proof}[Proof of Theorem~\ref{thm:comp-prof-swap}]
We begin by proving the two inequalities. To show that $\LinSwap(\bx,\by) \leq \CorrSwap(\bx, \by)$, let $\csp$ be the average CSP of the transcript $(\bx, \by)$, and let $\csp = \sum_{k}\lambda_{k}(x_{(k)} \otimes y_{(k)})$ be any decomposition of $\csp$ in the form \eqref{eq:csp-decomp}. Consider also any linear transformation $\psi$ that sends $\learnset$ to $\learnset$. Note that $\sum_{k}\lambda_{k}\Reg(x_{(k)}, y_{(k)}) \geq \sum_{k}\lambda_{k}(u_{L}(\psi(x_{(k)}), y_{(k)}) - u_{L}(x_{(k)}, y_{(k)})) = \frac{1}{T}\sum_{t}(u_{L}(\psi(x_t), y_t) - u_L(x_t, y_t))$, where the last equality follows by the linearity of $\psi$. Taking the minimum of the left-hand side of this inequality (over all decompositions of $\csp$) and the maximum of the right-hand side (over all linear endomorphisms $\psi$), we obtain that $\LinSwap(\bx, \by) \leq \CorrSwap(\bx, \by)$.

To show that $\CorrSwap(\bx, \by) \leq \PolySwap(\bx, \by)$, note that we can rewrite $\PolySwap(\bx, \by)$ in the form:

 \begin{align}
    \PolySwap(\bx, \by) &= \min_{\bxV} \max_{\pi: \learnvert \rightarrow \learnvert} \left(\sum_{t=1}^{T} u_L(\overline{\pi}(\xV_t), y_t) - u_L(x_t, y_t) \right) \notag \\
    &= T \cdot \min_{\bxV} \max_{\pi: \learnvert \rightarrow \learnvert} \left(\sum_{v \in \learnset} \alpha_v \left(u_L(\overline{\pi}(v), y_v) -  u_L(v, y_v)\right) \right) \label{posr_rewrite}
    \end{align}

\noindent
where the $\alpha_v$ and $y_v$ in \eqref{posr_rewrite} are obtained by expanding $\xV_t$ out as an explicit convex combination of vertices in $\learnvert$. In particular, every choice of per-round decompositions $\bxV$ in the outer minimum gives rise to a collection of $\alpha_v$ and $y_v$ satisfying $\alpha_v \geq 0$, $\sum_{v} \alpha_v = 1$, and $y_v \in \optset$. Since profile swap regret is the minimum over \emph{all} such decompositions (see \eqref{eq:vertex-csp}), $\CorrSwap(\bx, \by) \leq \PolySwap(\bx, \by)$.

We now prove the second part of the theorem, presenting families of transcripts that exhibit the above gaps. We provide an indirect proof that there exists an example where $\LinSwap(\bx, \by) = 0$ but $\CorrSwap(\bx, \by) = \Omega(T)$, by noticing that Theorem 8 of \cite{MMSSbayesian} provides an example of a transcript $(\bx, \by)$ in a polytope game with $\LinSwap(\bx, \by) = 0$ but that is manipulable. By Theorem~\ref{thm:poly_nonmanip}, it must be the case that $\CorrSwap(\bx, \by) > 0$. It is also possible to verify this example directly, but is computationally somewhat messy.

We now present an example where $\CorrSwap(\bx, \by) = 0$ but $\PolySwap(\bx, \by) = \Omega(T)$. This example is adapted from a counterexample in \cite{rubinstein2024strategizing} used to show that there exist transcripts with high polytope swap regret that were not manipulable. We will set $\learnset = \optset = \Delta_{2}^{2} = \{(a_1, a_2, b_1, b_2) \mid a_1 + a_2 = 1, b_1 + b_2 = 1, a_1, a_2, b_1, b_2 \in \Rset_{\geq 0}\}$ (this can be interpreted as a Bayesian game with two types and two actions for both players). For this set, we can write $\learnvert = \optvert = \{v_{11}, v_{12}, v_{21}, v_{22}\}$, where $v_{11} = (1, 0, 1, 0)$, $v_{12} = (1, 0, 0, 1)$, $v_{21} = (0, 1, 1, 0)$, and $v_{22} = (0, 1, 0, 1)$. We will let the learner's payoff be simply $u_L(x, y) = \langle x, y\rangle$. 

Now, consider the following transcript of play:

\begin{itemize}
    \item For $0 <  t \leq T/4$, $x_{t} = y_{t} = v_{11}$.
    \item For $T/4 < t \leq T/2$, $x_{t} = y_{t} = v_{12}$.
    \item For $T/2 < t \leq 3T/4$, $x_{t} = y_{t} = v_{21}$.
    \item For $3T/4 < t \leq T$, $x_{t} = v_{22}$ and $y_{t} = v_{11}$. \textbf{Note that $x_{t} \neq y_{t}$ for this period.}
\end{itemize}

We claim that $\PolySwap(\bx, \by) = \Omega(T)$ for this transcript of play. Indeed, note that since the learner plays an extreme point of $\learnset$ every round, any decomposition $\bxV$ must satisfy $\xV_{t} = x_t$. Now, the map $\pi$ which sends $v_{22}$ to $v_{11}$ (and fixes all other vertices) improves the learner's utility by $T/4 = \Omega(T)$.

On the other hand, we claim that $\CorrSwap(\bx, \by) = 0$. To see this, note that the CSP corresponding to this transcript can be written in the form

\begin{eqnarray*}
\csp &=& \frac{1}{4}\left(v_{11} \otimes v_{11} + v_{12} \otimes v_{12} + v_{21} \otimes v_{21} + v_{22} \otimes v_{11} \right)\\
&=& \frac{1}{4}\left((v_{11} + v_{22}) \otimes v_{11} + v_{12} \otimes v_{12} + v_{21} \otimes v_{21} \right) \\
&=& \frac{1}{4}\left((v_{12} + v_{21}) \otimes v_{11} + v_{12} \otimes v_{12} + v_{21} \otimes v_{21} \right) \\
&=& \frac{1}{4}\left(v_{12} \otimes (v_{12} + v_{11}) + v_{21} \otimes (v_{21} + v_{11}) \right). \\
\end{eqnarray*}

\noindent
Here, in the third line, we have used the fact that $v_{11} + v_{22} = v_{12} + v_{21}$. But now, note that $v_{12} \in \BR_{L}(v_{12} + v_{11})$, and $v_{21} \in \BR_{L}(v_{21} + v_{11})$. It follows that $\CorrSwap(\csp) = 0$, as desired.
\end{proof}

\subsection{Proof of Theorem~\ref{thm:equilibrium-gap}}

We will need the following two lemmas.

\begin{lemma} Consider any $v \in \cV(\learnset)$ and $\cV(\optset)_{0} \subseteq \cV(\optset)$ such that for all $w_{0} \in \cV(\optset)_{0}$, $(v,w_{0})$ has support zero in any NFCE. Then, if $(v,\hat{w}) \leq (\bar{v},\hat{w})$ for some $\bar{v} \in \cV(\learnset)$ and $\forall \hat{w} \in \cV(\optset) \setminus \cV(\optset)_{0}$, then for all $\hat{w} \in \cV(\optset)$ such that $(v,\hat{w}) < (\bar{v},\hat{w})$, $(v,\hat{w})$ has support zero in any NFCE.  \label{lem:elim_x}
\end{lemma}

\begin{proof}
Consider some $(v,\hat{w})$ pair that satisfies the conditions above, and assume for contradiction that this pair has support $\alpha > 0$ in some NFCE $\phi^{V}$. Conditioned on being recommended vertex $v$, Player $\cX$'s marginal distribution must be over only vertices in $\optvert \setminus \optvert_{0}$, as otherwise a vertex pair $(v,w_{0})$ would have nonzero support on some NFCE. Therefore, the normal-form swap regret of this NFCE is

\begin{align*}
    \NormSwap_{1}(\phi^{V}) & = \max_{\pi^{*}:\learnvert \mapsto \learnvert}\sum_{\forall v_i}\sum_{\forall w_j}\phi^{V}_{v_i,w_j}\left( u_{1}(\phi^{*}(v_i),w_j) - u_{1}(v_i,w_j) \right) \\
    & \geq \max_{\pi^{*}:\learnvert \mapsto \learnvert}\sum_{\forall w_j}\phi^{V}_{v,w_j}\left( u_{1}(\phi^{*}(v),w_j) - u_{1}(v,w_j) \right) \\
    & \geq \sum_{\forall w_j}\phi^{V}_{v,w_j}\left( u_{1}(\bar{v},w_j) - u_{1}(v,w_j) \right) \\
    & = \phi^{V}_{v,\hat{w}}\left( u_{1}(\bar{v},\hat{w}) - u_{1}(v,\hat{w}) \right) + \sum_{\forall w \ne \hat{w}}\phi^{V}_{v,w}\left( u_{1}(\bar{v},w) - u_{1}(v,w) \right) \\
     & \geq \phi^{V}_{v,\hat{w}}\left( u_{1}(\bar{v},\hat{w}) - u_{1}(v,\hat{w}) \right) \tag{As by assumption, for all $w \in \optvert$, either $u_1(\bar{v},w) > u_1(v,w)$ or $\phi^{V}_{v,w} = 0$} \\
       & = \alpha \left( u_{1}(\bar{v},\hat{w}) - u_{1}(v,\hat{w}) \right) > 0 \tag{By our assumption on the support and utility of $(v, \hat{w})$} \\
\end{align*}
Thus, the normal-form-swap regret of our NFCE $\phi^{V}$ is $> 0$, leading to a contradiction. 
\end{proof}

\begin{lemma} Consider any $w \in \cV(\optset)$ and $\cV(\learnset)_{0} \subseteq \cV(\learnset)$ such that for all $v_{0} \in \cV(\learnset)_{0}$, $(v_{0},w)$ has support zero in any NFCE. Then, if $(\hat{v},w) \leq (\hat{v},\bar{w})$ for some $\bar{w} \in \cV(\optset)$ and $\forall \hat{v} \in \cV(\learnset) \setminus \cV(\learnset)_{0}$, then for all $\hat{v} \in \cV(\learnset)$ such that $(\hat{v},w) < (\hat{v},\bar{w})$, $(\hat{v},w)$ has support zero in any NFCE. \label{lem:elim_y}
\end{lemma}

\begin{proof}
    The proof is symmetric to the proof of Lemma~\ref{lem:elim_x}.
\end{proof}

We can now prove Theorem~\ref{thm:equilibrium-gap}.

\begin{proof}[Proof of Theorem~\ref{thm:equilibrium-gap}]
First, note that by Theorem~\ref{thm:comp-prof-swap}, if $\csp = \Proj(\cspV)$, we must have that $\CorrSwap_{X}(\csp) \leq \NormSwap_{X}(\cspV)$ and $\CorrSwap_{Y}(\csp) \leq \NormSwap_{Y}(\cspV)$. Therefore, if $\cspV$ is a normal-form CE, it follows that $\csp$ is a profile CE.

We now construct a polytope game $G$ where $\csp$ is a profile CE, but no normal-form CE $\cspV$ has the property that $\Proj(\cspV) = \csp$.

The game is as follows. The two players share the same action set $\learnset = \optset = [0,1]^2 \times \{1\}$, and so $\learnvert = \optvert = \{[0,0,1],[0,1,1],[1,1,1],[1,0,1]\}$. The utility function $u_X$ of the $\cX$-player is\footnote{Throughout this proof, we regularly flatten elements of $\learnset \otimes \optset$ to elements of $\Rset^{9}$ for convenience of notation. In particular, we use the notation that $v \otimes w = [v_1w_1, v_1w_2, v_1w_3, v_2w_1, \dots, v_3w_3]$. We also represent bilinear functions by elements of $\Rset^{9}$ in a similar way; i.e., we define $u_X$ and $u_Y$ so that $u_X(\csp) = \langle u_x, \csp\rangle$.}

$$u_{X} = [-1, -1, 0.5, 0, 1, 0, -1, -1, -1].$$ 

The utility function of the $\cY$-player is

$$u_{Y} = [-0.5, 1, 1, -0.5, 0.5, 1, 0.5, -1, -0.5]$$

Finally the CSP $\csp$ is 

$$\csp = [0, 0.2, 0.4, 0.2, 0.4, 0.6, 0.4, 0.4, 1].$$

The remainder of the proof is a computation that justifies the correctness of this counterexample.

While the utilities are a direct function of the underlying dimensions of $\learnset$ and $\optset$, we can also compute the utilities of the game defined by $\learnvert$ and $\optvert$ by computing, for each extreme point $v$ and $w$ in the vertex game, $u_{X}(\pi(v \otimes w))$. Doing so, we can write out the utilities of Player $\cX$ and Player $\cY$ in the vertex game:

\begin{blockarray}{ccccc}
[0,0,1] & [0,1,1] & [1,1,1] & [1,0,1]  \\
\begin{block}{(cccc)c}
  -1 & -2 & -3 & -2 & [0,0,1]\\
  -1 &  -1 &  -2 & -2 & [0,1,1]\\
  -0.5 & -1.5 & -3.5 & -2.5 & [1,1,1]\\
  -0.5 & -2.5 & -4.5 & -2.5 & [1,0,1]\\
\end{block}
\end{blockarray}
Utility of Player $\cX$

\begin{blockarray}{ccccc}
[0,0,1] & [0,1,1] & [1,1,1] & [1,0,1]  \\
\begin{block}{(cccc)c}
  -0.5 & -1.5 & -1 & 0 & [0,0,1]\\
  0.5  & 0 &  0 &  0.5 & [0,1,1]\\
  1.5 & 2 &  1.5 & 1 & [1,1,1]\\
  0.5 &  0.5 & 0.5 & 0.5 & [1,0,1]\\
\end{block}
\end{blockarray}
Utility of Player $\cY$

First, we must prove that $\phi$ is a PCE. To prove this, it is sufficient to show that:
\begin{itemize}
    \item $\phi$ is a projection of a vertex CSP $\phi^{V}_{1}$ such that Player $\cX$ gets normal-form-swap regret of $0$ under $\phi^{V}_{1}$, and 
    \item $\phi$ is a projection of a vertex CSP $\phi^{V}_{2}$ such that Player $\cY$ gets normal-form-swap regret of $0$ under $\phi^{V}_{2}$
\end{itemize}

As normal-form swap regret upper-bounds profile swap regret, this proves that both players achieve $0$ PSR under $\phi$, and thus that $\phi$ is a PCE.  

For the first decomposition, consider the following vertex-game CSP $\phi^{V}_{1}$: 

\begin{blockarray}{ccccc}
[0,0,1] & [0,1,1] & [1,1,1] & [1,0,1]  \\
\begin{block}{(cccc)c}
  0.2 & 0 & 0 & 0.2 & [0,0,1]\\
  0  & 0 & 0.2 &  0 & [0,1,1]\\
  0.2 & 0.2 &  0 & 0 & [1,1,1]\\
  0 &  0 & 0 & 0 & [1,0,1]\\
\end{block}
\end{blockarray}

It is easy to see that Player $\cX$ attains NFSR of $0$ under $\phi^{V}_{1}$ by verifying that each of their actions in the vertex game is a best response against the conditional marginal distribution over Player $\cY$'s actions against it. Furthermore, 

\begin{align*}
\Proj(\cspV_{1}) &=
    \sum_{v \in \learnvert, w \in \optvert} \cspV_{1}(v, w) (v \otimes w) \\
    & = 0.2 \cdot ([0,0,1] \otimes [0,0,1]) + 0.2 \cdot ([0,0,1] \otimes [1,0,1]) + 0.2 \cdot ([0,1,1] \otimes [1,1,1]) \\ & + 0.2 \cdot ([1,1,1] \otimes [0,0,1]) + 0.2 \cdot ([1,1,1] \otimes [0,1,1]) \\
    & = 0.2 \cdot ([0,0,0,0,0,0,0,0,1]) + 0.2 \cdot ([0,0,0,0,0,0,1,0,1]) + 0.2 \cdot ([0,0,0,1,1,1,1,1,1]) \\ & + 0.2 \cdot ([0,0,1,0,0,1,0,0,1]) + 0.2 \cdot ([0,1,1,0,1,1,0,1,1]) \\
    =  & [0,0.2,0.4,0.2,0.4,0.6,0.4,0.4,1] = \phi
\end{align*}

For the second decomposition, consider the following vertex-game CSP $\phi^{V}_{2}$: 

\begin{blockarray}{ccccc}
[0,0,1] & [0,1,1] & [1,1,1] & [1,0,1]  \\
\begin{block}{(cccc)c}
  0  & 0  & 0 & 0.2 & [0,0,1]\\
  0  & 0.2& 0 & 0.2 & [0,1,1]\\
  0  & 0.2& 0 & 0 & [1,1,1]\\
  0.2& 0  & 0 & 0 & [1,0,1]\\
\end{block}
\end{blockarray}

It is again easy to see that Player $\cY$ attains NFSR of $0$ under $\phi^{V}_{2}$ by verifying that each of their actions in the vertex game is a best response against the conditional marginal distribution over Player $\cX$'s actions against it. Furthermore, we can verify that it is indeed a decomposition of $\phi$ by again computing its projection:

\begin{align*}
\Proj(\cspV_{2}) &=
    \sum_{v \in \learnvert, w \in \optvert} \cspV_{2}(v, w) (v \otimes w) \\
    & = 0.2 \cdot ([0,0,1] \otimes [1,0,1]) + 0.2 \cdot ([0,1,1] \otimes [0,1,1]) + 0.2 \cdot ([0,1,1] \otimes [1,0,1]) \\ & + 0.2 \cdot ([1,1,1] \otimes [0,1,1]) + 0.2 \cdot ([1,0,1] \otimes [0,0,1]) \\
    & = 0.2 \cdot ([0,0,0,0,0,0,1,0,1]) + 0.2 \cdot ([0,0,0,0,1,1,0,1,1]) + 0.2 \cdot ([0,0,0,1,0,1,1,0,1]) \\ & + 0.2 \cdot ([0,1,1,0,1,1,0,1,1]) + 0.2 \cdot ([0,0,1,0,0,0,0,0,1]) \\
    =  & [0,0.2,0.4,0.2,0.4,0.6,0.4,0.4,1] = \phi
\end{align*}

Therefore, $\phi = [0, 0.2, 0.4, 0.2, 0.4, 0.6, 0.4, 0.4, 1]$ is a PCE.

Next, we will show that there does not exist a NFCE $\sigma$ such that $\pi(\sigma) = \phi$. To prove this, we assume for contradiction that there does exist such an NFCE, and iteratively eliminate move pairs from its possible support until we find a contradiction. To do this, we will crucially make use of Lemmas~\ref{lem:elim_x} and~\ref{lem:elim_y}, which formalize the fact that not only can't a strictly dominated action be present in an NFCE, if there is a weakly dominated action\footnote{Here we use \emph{strictly} and \emph{weakly dominated} in a slightly different way than in when excluding weakly dominated actions from the game; here we say action $a$ strictly dominates $b$ when another $a$ is strictly better than $b$ for all actions of the opponent, and $a$ weakly dominates $b$ when it is weakly better for all actions, and strictly better for at least one.}, the action pairs where it is strictly worse than the dominating action also cannot be in any NFCE. These properties hold even if the domination is only over the support of the NFCE itself. 

It will be clarifying to keep track of this visually: we begin by allowing the possibility of all possible move pairs to be in the support of some NFCE, denoting this by a $?$ denoting the weight placed upon each move pair. As we prove that certain vertex action pairs cannot be in the support of any NFCE, we will replace these $?$s with $0$s. For simplicity, we will also re-label the vertices as $v_{1:4},w_{1:4}$. We begin by noting that since the first index of $\phi$ is $0$, the following entries are all zero:

\begin{blockarray}{ccccc}
$w_{1}$ & $w_{2}$ & $w_{3}$ & $w_{4}$ \\
\begin{block}{(cccc)c}
  ? & ? & ? & ? & $v_{1}$\\
  ? & ? & ? & ? & $v_{2}$\\
  ? & ? & $\mathbf{0}$ & $\mathbf{0}$ & $v_{3}$\\
  ? & ? & $\mathbf{0}$ & $\mathbf{0}$ & $v_{4}$\\
\end{block}
\end{blockarray}

First, note that $v_{1}$ is dominated for Player $\cX$ by $v_{2}$, and strictly dominated in two of the columns. Thus, by Lemma~\ref{lem:elim_x}, $(v_{1},w_{2})$ and $(v_{1},w_{3})$ must have support zero in any NFCE. 

Furthermore, $v_{4}$ is dominated for Player $\cX$ by $v_{3}$, and strictly dominated in two of the columns. Therefore, again by Lemma~\ref{lem:elim_x}, $(v_{4},w_{2})$ and $(v_{4},w_{3})$ must have support zero in any NFCE. Additionally, conditional on Player $\cX$ being recommended $v_{4}$ in a NFCE satisfying our constraints, Player $\cY$ must have support only on $w_{1}$ and $w_{2}$. Therefore by Lemma~\ref{lem:elim_y}, $(v_4,w_2)$ cannot be in the support either.

\begin{blockarray}{ccccc}
$w_{1}$ & $w_{2}$ & $w_{3}$ & $w_{4}$ \\
\begin{block}{(cccc)c}
  ? & $\mathbf{0}$ & $\mathbf{0}$ & ? & $v_{1}$\\
  ? & ? & ? & ? & $v_{2}$\\
  ? & ? & 0 & 0 & $v_{3}$\\
  ? & $\mathbf{0}$ & 0 & 0 & $v_{4}$\\
\end{block}
\end{blockarray}

Now, note that conditional on Player $\cY$ being recommended $w_{3}$, Player $\cX$ must have full support on $v_{2}$. Player $\cY$ would rather play $w_{1}$ in this case, so by Lemma~\ref{lem:elim_y} we can eliminate another vertex pair from the support. 

\begin{blockarray}{ccccc}
$w_{1}$ & $w_{2}$ & $w_{3}$ & $w_{4}$  \\
\begin{block}{(cccc)c}
  ? & 0 & 0 & ? & $v_{1}$\\
  ? & ? & $\mathbf{0}$ & ? & $v_{2}$\\
  ? & ? & 0 & 0 & $v_{3}$\\
  ? & 0 & 0 & 0 & $v_{4}$\\
\end{block}
\end{blockarray}

From here, we can label the remaining vertex pairs in order to determine additional constraints: 

\begin{blockarray}{ccccc}
$w_{1}$ & $w_{2}$ & $w_{3}$ & $w_{4}$ \\
\begin{block}{(cccc)c}
  a & 0 & 0 & b & $v_{1}$\\
  c & d & 0 & e & $v_{2}$\\
  f & g & 0 & 0 & $v_{3}$\\
  h & 0 & 0 & 0 & $v_{4}$\\
\end{block}
\end{blockarray}

Recall that $\phi[4] = 0.2$, and further that $\phi[4]$ is the sum of the values at the intersection of the leftmost two columns and the middle rwo rows. Therefore, this tells us that $e = 0.2$. Furthermore, $\phi[2] = 0.2$ as well, and this value represents the sum of the four values in the intersection of the middle two columns and the lowest two rows. Thus, $g = 0.2$. Also, $\phi[5] = 0.4$, representing the central four values in the matrix. Thus, $d + g = 0.4$, and therefore $d = 0.2$.

\begin{blockarray}{ccccc}
$w_{1}$ & $w_{2}$ & $w_{3}$ & $w_{4}$  \\
\begin{block}{(cccc)c}
  a & 0 & 0 & b & $v_{1}$\\
  c & $\mathbf{0.2}$ & 0 & $\mathbf{0.2}$ & $v_{2}$\\
  f & $\mathbf{0.2}$ & 0 & 0 & $v_{3}$\\
  h & 0 & 0 & 0 & $v_{4}$\\
\end{block}
\end{blockarray}

Also, $\phi[6] = 0.6$, which represents the middle two rows of the matrix. Therefore, $c$ and $f$ must both equal $0$:

\begin{blockarray}{ccccc}
$w_{1}$ & $w_{2}$ & $w_{3}$ & $w_{4}$   \\
\begin{block}{(cccc)c}
  a & 0 & 0 & b & $v_{1}$\\
  $\mathbf{0}$ & 0.2 & 0 & 0.2 & $v_{2}$\\
  $\mathbf{0}$ & 0.2 & 0 & 0 & $v_{3}$\\
  h & 0 & 0 & 0 & $v_{4}$\\
\end{block}
\end{blockarray}

And $\phi[3] = 0.4$, which represents the bottom two rows. There is only one variable remaining in these rows we can set, and thus $h = 0.2$:

\begin{blockarray}{ccccc}
$w_{1}$ & $w_{2}$ & $w_{3}$ & $w_{4}$  \\
\begin{block}{(cccc)c}
  a & 0 & 0 & b & $v_{1}$\\
  0 & 0.2 & 0 & 0.2 & $v_{2}$\\
  0 & 0.2 & 0 & 0 & $v_{3}$\\
  $\mathbf{0.2}$ & 0 & 0 & 0 & $v_{4}$\\
\end{block}
\end{blockarray}

Next, $\phi[7] = 0.4$, which represents the leftmost two columns. Thus, $b = 0.2$:

\begin{blockarray}{ccccc}
$w_{1}$ & $w_{2}$ & $w_{3}$ & $w_{4}$   \\
\begin{block}{(cccc)c}
  a & 0 & 0 & $\mathbf{0.2}$ & $v_{1}$\\
  0 & 0.2 & 0 & 0.2 & $v_{2}$\\
  0 & 0.2 & 0 & 0 & $v_{3}$\\
  0.2 & 0 & 0 & 0 & $v_{4}$\\
\end{block}
\end{blockarray}

Finally, given that these values must all sum to $1$, this leaves us with a unique solution:

\begin{blockarray}{ccccc}
$w_{1}$ & $w_{2}$ & $w_{3}$ & $w_{4}$   \\
\begin{block}{(cccc)c}
  $\mathbf{0}$ & 0 & 0 & 0.2 & $v_{1}$\\
  0 & 0.2 & 0 & 0.2 & $v_{2}$\\
  0 & 0.2 & 0 & 0 & $v_{3}$\\
  0.2 & 0 & 0 & 0 & $v_{4}$\\
\end{block}
\end{blockarray}

However, this is \emph{not} an NFCE. Conditioned on Player $\cX$ playing $v_{3}$, Player $\cX$'s marginal distribution is full support on $w_{2}$. But in this case, Player $\cX$ has normal-form-swap regret to playing $v_{2}$. This leads to a contradiction, and therefore $\phi$ is not the projection of any NFCE.

Furthermore, we verified via a linear program that there does not exist any other $\phi^{*}$ which is the projection of some NFCE, and has the same utility profile as $\phi$. 
\end{proof}

\end{document}